\documentclass{lmcs}
\pdfoutput=1

\usepackage{lastpage}
\lmcsdoi{15}{3}{14}
\lmcsheading{}{\pageref{LastPage}}{}{}%
{Dec.~25,~2017}{Aug.~08,~2019}{}


\usepackage{amsmath,amssymb}
\usepackage{graphics}
\usepackage{relsize}
\usepackage{wrapfig}
\usepackage{mathpartir}
\usepackage{mathtools}

\usepackage{rotating,color}

\definecolor{citecolor}{rgb}{0.0,0.4,0.0}
\definecolor{urlcolor}{rgb}{0.0,0.0,0.4}
\definecolor{linkcolor}{rgb}{0.0,0.0,0.4}

\usepackage[all]{xy}

\theoremstyle{plain}
\newtheorem{lemma}[thm]{Lemma}
\newtheorem{proposition}[thm]{Proposition}
\newtheorem{theorem}[thm]{Theorem}
\newtheorem{corollary}[thm]{Corollary}
\theoremstyle{definition}
\newtheorem{definition}[thm]{Definition}
\newtheorem{example}[thm]{Example}
\newtheorem{remark}[thm]{Remark}

\newcommand{\Id}{\mathsf{Id}}
\newcommand{\id}{\mathsf{id}}
\newcommand{\A}{\mathcal{A}}

\newcommand{\C}{\mathcal{C}}
\newcommand{\D}{\mathcal{D}}
\newcommand{\E}{\mathcal{E}}
\newcommand{\F}{\mathcal{F}}
\newcommand{\K}{\mathcal{K}}
\newcommand{\DL}[1]{\mathsf{DL}(#1)} % chktex 36
\newcommand{\causal}[1]{\mathsf{causal}(#1)} % chktex 36
\newcommand{\pow}{\mathcal{P}}
\newcommand{\powf}{\mathcal{P}_f}
\newcommand{\pfs}{P}
\newcommand{\real}{\mathbb{R}}

\newcommand{\op}{\mathsf{op}}

\newcommand{\even}{\mathsf{even}}
\newcommand{\double}{\mathsf{double}}

\newcommand{\Ran}[2]{\mathsf{Ran}_{#1}(#2)} % chktex 36
\newcommand{\fin}[1]{\bar{#1}}
\newcommand{\cod}[1]{\mathsf{C}_{#1}}
\newcommand{\cf}[1]{\cod{\fin{B}}}
\newcommand{\kfab}{\Ran{\fin{A}}{\fin{B}}}

\newcommand{\Rel}{\mathsf{Rel}}
\newcommand{\ctx}{\mathsf{ctx}}
\newcommand{\rfl}{\mathsf{rfl}}

\newcommand{\sol}[1]{#1^\dagger}

\newcommand{\Ord}{\mathsf{Ord}}
\newcommand{\Set}{\mathsf{Set}}

\newcommand\eqdef\triangleq%
\newcommand\RR{\mathbb{R}}

\newcommand\paren[1]{\left({#1}\right)}
\newcommand\tuple\paren%

\begin{document}

\title[Companions, Codensity and Causality]{Companions, Codensity and Causality}
\titlecomment{{\lsuper*}
    Extended version of the paper
    published by Springer in Proc. FoSSaCS'17.
    The research leading to these results has received funding from
    the European Research Council (FP7/2007--2013, grant agreement
    nr.~320571; and H2020, grant agreement nr.~678157); as well as
    from the LABEX MILYON (ANR-10-LABX-0070, ANR-11-IDEX-0007).}

\author[D.~Pous]{Damien Pous\rsuper{a}}	%required
\address{\lsuper{a}Univ Lyon, CNRS, ENS de Lyon, UCB Lyon 1, LIP, France}	%required
\email{Damien.Pous@ens-lyon.fr}  %optional
%\thanks{thanks 1, optional.}	%optional

\author[J.~Rot]{Jurriaan Rot\rsuper{b}}	%optional
\address{\lsuper{b}Radboud University, Nijmegen}	%optional
\email{jrot@cs.ru.nl}  %optional
%\thanks{thanks 2, optional.}	%optional

\begin{abstract}
  In the context of abstract coinduction in complete lattices, the
  notion of compatible function makes it possible to introduce
  enhancements of the coinduction proof principle. The largest
  compatible function, called the companion, subsumes most
  enhancements and has been proved to enjoy many good properties. Here
  we move to universal coalgebra, where the corresponding notion is
  that of a \emph{final} distributive law. We show that when it exists,
  the final distributive law is a monad, and that it coincides with
  the codensity monad of the final sequence of the given functor. On
  sets, we moreover characterise this codensity monad using a new
  abstract notion of causality. In particular, we recover the fact
  that on streams, the functions definable by a distributive law or
  GSOS specification are precisely the causal functions. Going back to
  enhancements of the coinductive proof principle, we finally obtain
  that any causal function gives rise to a valid up-to-context
  technique.
\end{abstract}

\maketitle

\section{Introduction}

Coinduction has been widely studied in computer science since Milner's
work on CCS~\cite{Milner89}. In concurrency theory, it is usually
exploited to define behavioural equivalences or preorders on processes
and to obtain powerful proof principles.
% : bisimilarity is the largest
% bisimulation, so that to prove two processes equivalent it suffices
% to find a bisimulation that relates them.
Coinduction can also be used for programming languages, to define and
manipulate infinite data-structures like streams or potentially
infinite trees. For instance, streams can be defined using systems of
differential equations~\cite{Rutten05}. In particular, pointwise
addition of two streams $x,y$ can be defined by the following
equations, where $x_0$ and $x'$ respectively denote the head and the
tail of a stream $x$.
\begin{equation}
  \label{eq:plus1}
  \begin{aligned}
    {(x \oplus y)}_0 &= x_0 + y_0\\
    (x \oplus y)' &= x' \oplus y'
  \end{aligned}
\end{equation}

Coinduction as a proof principle for concurrent systems can be nicely
presented at the abstract level of complete
lattices~\cite{pous:aplas07:clut,pous:dsbook11}: bisimilarity is the
greatest fixpoint of a monotone function on the complete lattice of
binary relations. In contrast, coinduction as a tool to manipulate
infinite data-structures requires one more step to be presented
abstractly: moving to universal coalgebra~\cite{Jacobs:coalg}. For
instance, streams are the carrier of the final coalgebra of an
endofunctor on the category $\Set$ of sets and functions, and simple systems of differential equations
are just plain coalgebras. One usually
distinguishes between coinduction as a tool to prove properties, and \emph{corecursion} or
\emph{coiteration} as a tool to define functions. The theory we
develop in the present work encompasses both in a uniform way, so that
we do not emphasise their differences and just use the word
\emph{coinduction}.

In both cases one frequently needs enhancements of the coinduction
principle~\cite{San98MFCS,SW01}. Indeed, rather than working with
plain bisimulations, which can be rather large, one often uses
``bisimulations up-to'', which are not proper bisimulations but are
nevertheless contained in
bisimilarity~\cite{MilnerPW92,Abr90,AbadiG98,fournet00asynchronous,JeffreyR04,Compcert-CACM,SevcikVNJS13}.
The situation with infinite data-structures is similar. For instance, defining the
shuffle product on streams is typically done using equations of
the following shape,
\begin{equation}
  \label{eq:times1}
  \begin{aligned}
    {(x \otimes y)}_0 &= x_0 \times y_0\\
    (x \otimes y)' &= x \otimes y' ~\oplus~ x' \otimes y
  \end{aligned}
\end{equation}
which fall out of the scope of plain coinduction due to the call to
pointwise addition~\cite{Rutten05,HansenKR16}.

Enhancements of the bisimulation proof method have been introduced by
Milner from the beginning~\cite{Milner89}, and further studied by
Sangiorgi~\cite{San98MFCS,SW01} and then by the first
author~\cite{pous:aplas07:clut,pous:dsbook11}. Let us recall the
standard formulation of coinduction in complete lattices: by
Knaster-Tarski's theorem~\cite{Kna28,Tarski55}, any monotone function
$b$ on a complete lattice admits a greatest fixpoint $\nu b$ that
satisfies the following \emph{coinduction principle}:
\begin{align}
  \label{eq:coind}
  \inferrule*[right=coinduction]{x \leq y \leq b(y)}{x \leq \nu b}
\end{align}
In words, to prove that some point $x$ is below the greatest fixpoint,
it suffices to exhibit a point $y$ above $x$ which is an
\emph{invariant}, i.e., a post-fixpoint of $b$. Enhancements, or up-to
techniques, make it possible to alleviate the second requirement:
instead of working with post-fixpoints of $b$, one might use
post-fixpoints of $b\circ f$, for carefully chosen functions $f$:
\begin{align}
  \label{eq:coindupto}
  \inferrule*[right=coinduction up to $f$]{x \leq y \leq b(f(y))}{x \leq \nu b}
\end{align}
Taking inspiration from the work of Hur et al.~\cite{paco13}, the first author
recently proposed to systematically use for $f$ the largest
\emph{compatible} function~\cite{pous:lics16:cawu}, i.e., the largest
function $t$ such that $t\circ b\leq b \circ t$. Such a function
always exists and is called the \emph{companion}. It enjoys many good
properties, the most important one possibly being that it is a closure
operator: $t\circ t=t$. Parrow and Weber characterised it
extensionally in terms of the final sequence of the function
$b$~\cite{ParrowWeber16,pous:lics16:cawu}:
\begin{align}
  \label{eq:pw16}
  t:x\mapsto \bigwedge_{x\leq b_\alpha} b_\alpha &&\text{where }
  \begin{cases}
    b_\lambda \eqdef \bigwedge_{\alpha<\lambda} b_\alpha & \text{for limit ordinals}\\
    b_{\alpha+1} \eqdef b(b_\alpha) & \text{for successor ordinals}
  \end{cases}
\end{align}

In the present paper, we give a categorical account of these ideas,
generalising them from complete lattices to universal coalgebra, in
order to encompass important instances of coinduction such as solving
systems of equations on infinite data-structures.

\medskip

Let us first be more precise about our example on streams. We consider
there the $\Set$ functor $BX=\RR\times X$, whose final coalgebra is
the set $\RR^\omega$ of streams over the reals. This means that any
$B$-coalgebra $\tuple{X,f}$ defines a function from $X$ to
streams. Take for instance the following coalgebra over the
two-elements set $2=\{0,1\}$:
% \begin{align*}
%   \alpha:~&2 \to B2\\
%           &
            $0 \mapsto \tuple{0.3,\,1}$,
%\\
          % &
          $1 \mapsto \tuple{0.7,\,0}$.
% \end{align*}
This coalgebra can be seen as a system of two equations, whose unique
solution is a function from $2$ to $\RR^\omega$, i.e, two streams, where
the first has value $0.3$ at all even positions and $0.7$ at all
odd positions.

In a similar manner, one can define binary operations on streams by
considering coalgebras whose carrier consists of pairs of streams. For
instance, the previous system of equations characterising pointwise
addition~(\ref{eq:plus1}) is faithfully represented by the following
coalgebra:
\begin{equation}
  \label{eq:coalg-sum}
  \begin{aligned}
  &{(\RR^\omega)}^2 \to B({(\RR^\omega)}^2)\\
  &\tuple{x,y} \mapsto \tuple{x_0+y_0,~\tuple{x',y'}}
\end{aligned}
\end{equation}
Unfortunately, as explained above, systems of equations defining
operations like shuffle product (\ref{eq:times1}) cannot be
represented easily in this way: we would need to call pointwise
addition on streams that are not yet fully defined.

Instead, let $S$ be the functor $SX=X^2$ so that pointwise stream
addition can be seen as an $S$-algebra. Equations~(\ref{eq:times1})
can be represented by the following $BS$-coalgebra.
\begin{equation}\label{eq:coalg-shuffle}
\begin{aligned}
  &{(\RR^\omega)}^2 \to BS({(\RR^\omega)}^2)\\
  &\tuple{x,y} \mapsto \tuple{x_0\times y_0,~\tuple{\tuple{x,y'},\,\tuple{x',y}}}
\end{aligned}
\end{equation}
The inner pairs $\tuple{x,y'}$ and $\tuple{x',y}$ correspond to the
corecursive calls, and thus to the shuffle products $x\otimes y'$ and
$x'\otimes y$; in contrast, the intermediate pair
$\tuple{\tuple{x,y'},\,\tuple{x',y}}$ corresponds to a call to the
algebra on $S$, i.e., in this case, pointwise addition.

We should now explain in which sense and under which conditions such a
$BS$-coalgebra gives rise to an operation on the final $B$-coalgebra.
A preliminary step actually consists in proposing the following
principle (Definition~\ref{def:validity-coind}).

\begin{quote}\it
  Let $B$ be an endofunctor with a final $B$-coalgebra $(Z,\zeta)$,
  let $F$ be a functor, and let $\alpha\colon FZ\rightarrow Z$ be an
  $F$-algebra on the final $B$-coalgebra. \emph{Coinduction up to the
    algebra $\alpha$ is valid} if for every $BF$-coalgebra $(X,g)$,
  there exists a unique morphism $\sol g\colon X\rightarrow Z$ making
  the following diagram commute.
\begin{equation}
  \label{eq:intro:coindupto}
  \vcenter{
  \xymatrix{
  X \ar[d]_{g} \ar[rr]^{\sol{g}} % chktex 3
  & & Z \ar[d]^{\zeta} \\ % chktex 3
  BFX \ar[r]_-{BF\sol{g}} % chktex 3
  & BFZ \ar[r]_{B\alpha} & BZ } % chktex 3
  }
\end{equation}
\end{quote}
Intuitively, $\sol g$ gives the solution of the system of equations
represented by $g$: it interprets the variables $(X)$ into the
denotational space of behaviours $(Z)$; the above diagram ensures that
the equations are satisfied when using the algebra $\alpha$ to
interpret the $F$-part of the equations. For the previous
example~\eqref{eq:coalg-shuffle}, one would take $F=S$ and pointwise
addition for $\alpha$. This notion of validity improves over the
notion of soundness we proposed before~\cite{bppr:acta:16} in that it
fully specifies in which sense the equations are solved, by mentioning
explicitly the algebra which is used on the final coalgebra.

In the literature on universal coalgebra, one would typically prove
such an enhanced coinduction principle by using distributive
laws. Typically, this principle holds if $F$ is a monad and if there
exists a distributive law $\lambda \colon FB\Rightarrow BF$ of a monad
$F$ over $B$ (e.g.,~\cite{Bartels04,LenisaPW00,UustaluVP01,Jacobs06,MiliusMS13}).
The proof relies on the so-called generalised powerset construction~\cite{SilvaBBR10} and the
$F$-algebra $\alpha$ is the one canonically generated by
$\lambda$. This precisely amounts to using an up-to technique. Such a
use of distributive laws is actually rather standard in operational
semantics~\cite{TuriP97,Bartels04,Klin11}; they properly generalise
the notion of compatible function. In order to
follow~\cite{pous:lics16:cawu}, we thus focus on the largest
distributive law.

\medskip

Our first contribution consists in showing that if a functor $B$
admits a final distributive law (called the companion), then (1) this
distributive law is that of a monad $T$ over $B$, and thus (2)
coinduction up to its associated algebra is valid
(Section~\ref{sec:companion}). In complete lattices, this corresponds
to the facts that the companion is a closure operator and that it can
be used as an up-to technique.

\medskip

Then we move to conditions under which the companion exists. We start
from the \emph{final sequence} of the functor $B$, which is commonly
used to obtain the existence of a final
coalgebra~\cite{adamek1974free,barr}, and we show that the companion
actually coincides with the \emph{codensity monad} of this sequence,
provided that this codensity monad exists and is preserved by $B$
(Theorem~\ref{thm:companion-codensity}). Those conditions are
satisfied by all polynomial functors. This link with the final
sequence of the functor makes it possible to recover Parrow and
Weber's
characterisation % of the companion in the complete lattice case
(Equation~(\ref{eq:pw16})).

We can go even further for $\omega$-continuous endofunctors on $\Set$:
the codensity monad of the final sequence can be characterised in
terms of a new abstract notion of \emph{causal algebra}
(Definition~\ref{def:causal-alg}). On streams, this notion coincides
with the standard notion of causality~\cite{HansenKR16}: causal
algebras (on streams) correspond to operations such that the $n$-th
value of the result only depends on the $n$-th first values of the
arguments. For instance, pointwise addition and shuffle product are
causal algebras for the functor $SX=X^2$.

\medskip

These two characterisations of the companion in terms of the codensity
monad and in terms of causal algebras are the key theorems of the
present paper. We study some of their consequences in
Section~\ref{sec:comp-set}. These apply to all polynomial functors.

First, coinduction up to, as presented in~\eqref{eq:intro:coindupto},
is valid for every causal algebra for a functor $F$. Such a technique
makes it possible to define shuffle product~(\ref{eq:times1}) in a
streamlined way, without mentioning any distributive law. In the very same
way, with the functor $BX=2\times X^A$ for deterministic automata, we
immediately obtain the semantics of non-deterministic automata and
context-free grammars using simple causal algebras on formal languages
(Example~\ref{ex:cfg})---distributive laws are now hidden from the
end-user.

Second, we obtain that algebras on the final coalgebra are causal if
and only if they can be defined by a distributive law. Similar results
were known to hold for streams~\cite{HansenKR16} and
languages~\cite{RotBR16}. Our characterisation is more abstract and
less syntactic; the precise relationship between those results remains
to be studied.

Third, we can combine our results with some recent
work~\cite{bppr:csllics14:fibupto} where we rely on (bi)fibrations to % chktex 36
lift distributive laws on systems (e.g., automata) to obtain
up-to techniques for coinductive predicates or relations on those
systems (e.g., language equivalence). Doing
so, we obtain that every causal algebra gives rise to a valid up-to
context technique (Section~\ref{ssec:fibrations}). For instance,
bisimulation up to pointwise addition and shuffle product is a
valid technique for proving stream equalities coinductively.

\medskip

In Section~\ref{sec:gsos} we provide an expressivity result:
while abstract GSOS specifications~\cite{TuriP97} seem more expressive
than plain distributive laws, we show that this is actually not the
case: any algebra obtained from an abstract GSOS specification can
actually be defined from a plain distributive law.

Our main results on the construction of a companion through the final sequence
apply to polynomial functors. In Section~\ref{sec:pres-pow} we show a negative result:
the finite powerset functor (which is not polynomial) does not satisfy
the premises of our results, and hence falls outside its scope.
We conclude the paper with related work (Section~\ref{sec:rw}) and future work (Section~\ref{sec:fw}).

% Finally, A standard alternative in order to define operations on a final
% coalgebra consists in using \emph{abstract GSOS}
% specifications~\cite{TuriP97}, i.e., distributive laws of the form
% $\lambda \colon F(B \times \Id) \Rightarrow BF^*$, where $F^*$ is the
% free monad for $F$. Turi and Plotkin use this complicated type because
% it is apparently more expressive. (E.g., many process calculi can be
% defined naturally in this way but not using plain distributive laws of
% the form $FB\Rightarrow BF$.) Using the codensity monad, we

\paragraph{Acknowledgments.} We are grateful to Henning Basold,
Filippo Bonchi, Tom van Bussel, Bart Jacobs, Joshua Moerman, Daniela Petri{\c{s}}an,
and Jan Rutten for valuable discussions and comments.

\section{Preliminaries}\label{sec:prelim}

Let $\C$ be a category.
A \emph{coalgebra} for a functor $B \colon \C \rightarrow \C$ is a
pair $(X,f)$ where $X$ is an object in $\C$ and
$f\colon X \rightarrow BX$ a morphism.  A coalgebra homomorphism from
$(X,f)$ to $(Y,g)$ is a $\C$-morphism $h \colon X \rightarrow Y$ such
that $g \circ h = Bh \circ f$.
% The category of coalgebras for a functor $B$ is denoted by
% $\coalg{B}$.
A coalgebra $(Z,\zeta)$ is called \emph{final} if it is final in the
category of coalgebras%$\coalg{B}$
, i.e., for every coalgebra $(X,f)$ there exists a unique coalgebra
morphism from $(X,f)$ to $(Z,\zeta)$.

%% DAMIEN: unecessary, the two examples come naturally in the next subsection
% There are two running examples of coalgebras for $\Set$ functors,
% but see~\cite{Jacobs:coalg,Rutten00} for many more.  A \emph{stream
% system} is a coalgebra for the functor
% $B \colon \Set \rightarrow \Set$ given by $BX = A \times X$, where
% $A$ is a fixed set. A \emph{deterministic automaton} is a coalgebra
% for the functor $X \mapsto 2 \times X^A$.  For details on the final
% coalgebra of stream systems and deterministic automata, see
% Example~\ref{ex:finseq-streams} and
% Example~\ref{ex:finseq-automata}.

An \emph{algebra} for a functor $F \colon \C \rightarrow \C$ is defined dually to
a coalgebra, i.e., it is a pair $(X, a)$ where $a \colon FX \rightarrow X$,
and an algebra morphism from $(X,a)$ to $(Y,b)$ is a morphism $h \colon X \rightarrow Y$
such that $h \circ a = b \circ Fh$.
%The category of $L$-algebras is denoted by $\alg{L}$.
%An algebra is called \emph{initial} if it is an initial object in $\alg{L}$.

A \emph{monad} is a triple $(T, \eta, \mu)$ where
$T \colon \C \rightarrow \C$ is a functor, and
$\eta \colon \Id \Rightarrow T$ and $\mu \colon TT \Rightarrow T$ are
natural transformations called \emph{unit} and \emph{multiplication}
respectively, such that
%the following diagrams commute:
%\begin{equation}\label{eq:prel:monad}
%\begin{gathered}
%\xymatrix{
%	T \ar@{=>}[r]^-{\eta T} \ar@{=}[dr]
%		& TT \ar@{=>}[d]^-\mu
%		& T \ar@{=}[dl] \ar@{=>}[l]_-{T\eta} \\
%	  & T &
%}
%\qquad
%\xymatrix{
%	TTT \ar@{=>}[d]_-{\mu T} \ar@{=>}[r]^-{T\mu}
%		& TT \ar@{=>}[d]^-\mu\\
%	TT \ar@{=>}[r]_-\mu
%		& T
%}
%\end{gathered}
%\end{equation}
$\mu \circ T \eta = \id = \mu \circ \eta T$ and $\mu \circ \mu T = \mu \circ T \mu$.

By $\pow \colon \Set \rightarrow \Set$ we denote the (covariant) powerset functor,
and by $\powf \colon \Set \rightarrow \Set$ the finite powerset functor,
defined by $\powf(X) = \{S \subseteq X \mid S \text{ is finite}\}$.

\subsection{Final sequence} %\label{sec:fin-seq}

Let $B \colon \C \rightarrow \C$ be an endofunctor on a complete
category $\C$.  The \emph{final sequence} is the unique
ordinal-indexed sequence defined by $B_0 = 1$ (the final object of
$\C$), $B_{i+1} = BB_i$ and $B_j = \lim_{i < j} B_i$ for a limit
ordinal $j$, with
%	\begin{align*}
%   W_0 = \term \qquad
%		W_{i+1} = BW_i  \qquad
%		W_j = \lim_{i < j} W_i ~\text{(for $j$ a limit ordinal)}.
%	\end{align*}
connecting morphisms $B_{j,i} \colon B_j \rightarrow B_i$ for all
$i \leq j$, satisfying $B_{i,i} = \id$, $B_{j+1,i+1} = BB_{j,i}$ and
if $j$ is a limit ordinal then ${(B_{j,i})}_{i < j}$ is a limit cone.

The final sequence is a standard tool for constructing final
coalgebras: if there exists an ordinal $k$ such that $B_{k+1,k}$ is an
isomorphism, then $B_{k+1,k}^{-1} \colon B_k \rightarrow BB_k$ is a
final $B$-coalgebra~\cite[Theorem 1.3]{barr} (and dually for initial
algebras~\cite{adamek1974free}).  In the sequel, we shall sometimes
present it as a functor $\fin{B} \colon \Ord^\op \rightarrow \C$,
given by $\fin{B}(i) = B_i$ and $\fin{B}(j,i) = B_{j,i}$.

\begin{example}\label{ex:finseq-streams}
  Consider the functor $B \colon \Set \rightarrow \Set$ given by
  $BX = \real \times X$, whose coalgebras are stream systems over the real numbers.  Then
  $B_0 = 1$ and $B_{i+1} = \real \times B_i$ for $0 < i < \omega$. Hence,
  for $i < \omega$, $B_i$ is the set of all finite lists over $\real$ of
  length $i$. The limit $B_\omega$ consists of the set of all streams
  over $\real$.  For each $i,j$ with $i \leq j$, the connecting map
  $B_{j,i}$ maps a stream (if $j=\omega$) or a list (if $j<\omega$) to
  the prefix of length $i$.  The set $B_\omega$ of streams is a final
  $B$-coalgebra.
\end{example}

\begin{example}\label{ex:finseq-automata}
  For the $\Set$ functor $BX = 2 \times X^A$ whose coalgebras are
  deterministic automata over $A$, $B_i$ is (isomorphic to) the set of
  languages of words over $A$ with length below $i$. In particular,
  $B_\omega = \pow(A^*)$ is the set of all languages, and it is a
  final $B$-coalgebra.
\end{example}

\begin{example}\label{ex:finseq-plusone}
  For the $\Set$ functor $BX = X + 1$, $B_i$ is (isomorphic to) the
  set $\{0,1,2,\ldots,i\}$. For $i < j$, the projection $B_{j,i}$ is
  given by
  $B_{j,i}(k) = \begin{cases} k & \text{ if } k \leq i \\ i &
    \text{otherwise} \end{cases}$.
\end{example}

\begin{example}
  Let $\C$ be a complete lattice, seen as a poset category. An
  endofunctor is just a monotone function $b$; the final coalgebra is
  its greatest fixed point $\nu b$; the final sequence corresponds to
  the sequence defined in Equation~(\ref{eq:pw16}), and we recover
  Kleene's fixpoint theorem:
  \begin{align*}
    \nu b = \bigwedge_\alpha b_\alpha
  \end{align*}
\end{example}

A functor $B \colon \C \rightarrow \C$ is called
\emph{($\omega$)-continuous} if it preserves limits of
$\omega^{\op}$-chains. For such a functor, $B_\omega$ is the carrier
of a final $B$-coalgebra.  The functors of stream systems and automata
in the above examples are both $\omega$-continuous.

%\todo{DAM: at some point we wanted to announce the generalisation to A
%  B in the introduction, but I think the way you implemented this
%  generalisation in the sequel does not require it}

\section{Distributive laws and coinduction up-to}%
\label{sec:dl}

We define an abstract coinduction-up-to principle with respect to algebras on the final coalgebra,
and show its validity under the assumption of a distributive law.
\begin{definition}\label{def:validity-coind}
  Let $B,F \colon \C \rightarrow \C$ be functors such that $B$ has a
  final coalgebra $(Z,\zeta)$, and let
  $\alpha \colon FZ \rightarrow Z$ be an algebra on it.  We say that
  \emph{coinduction up to $\alpha$ is valid} if for every coalgebra
  $g \colon X \rightarrow BFX$ there exists a unique arrow
  $\sol{g} \colon X \rightarrow Z$ such that the following diagram
  commutes.
  \[
  \xymatrix{
    X \ar[d]_{g} \ar[rr]^{\sol{g}} % chktex 3
    & & Z \ar[d]^{\zeta} \\ % chktex 3
    BFX \ar[r]_-{BF\sol{g}} % chktex 3
    & BFZ \ar[r]_{B\alpha} % chktex 3
    & BZ
  }
  \]
\end{definition}
%A distributive law $\lambda$ of a monad over a functor allows one to
%strengthen the coinduction principle obtained by finality, as observed
%in~\cite{Bartels04} (specifically its Corollary 4.3.6), where it is
%called \emph{$\lambda$-coiteration}.  This principle allows one to
%solve (co)recursive equations, see, e.g., loc.\ cit.\ and
%\cite{Jacobs06,MiliusMS13}.

For $b,f \colon L \rightarrow L$ monotone functions on a complete
lattice, the existence of an algebra $\alpha$ as in
Definition~\ref{def:validity-coind} means that $f$ preserves the
greatest fixed point: $f(\nu b) \leq \nu b$.  Validity in the sense of
Definition~\ref{eq:coindupto} amounts to the validity of coinduction
up to $f$ (Equation~\ref{eq:coindupto}, also called soundness) with
the implicit requirement that $f(\nu b) \leq \nu b$. The latter is not
required by soundness and does not follow from
it~\cite[page~53]{pous:phd}.

As mentioned in the Introduction, this definition also differs from
the notion of soundness we proposed in~\cite[Section~3]{bppr:acta:16}:
there we were focusing on coinduction up to a functor $F$ rather than
up to an algebra for a functor, and we were asking for the existence
of an appropriate functor $G$ from the category of $BF$-coalgebras to
that of $B$-coalgebras. The problem with soundness there is that the
sense in which the system of equations is solved is specified only
through the functor $G$, and thus through the proof of soundness. In
contrast, the present definition of validity is more fine-grained and
fully characterises the way solutions behave.

\begin{example}\label{ex:shuffle-coin-up-to}
  Recall, from Equation~\ref{eq:coalg-shuffle}, how we presented the
  definition of the shuffle product
  $\otimes \colon {(\real^\omega)}^2 \rightarrow \real^\omega$ as a
  $BS$-coalgebra, for $SX = X^2$ and $BX = \real \times X$.  Let
  $\oplus \colon {(\real^\omega)}^2 \rightarrow \real^\omega$ be
  pointwise addition of streams.  Coinduction up to $\oplus$ is valid,
  and the shuffle product arises as the unique arrow from the above
  coalgebra to the final $B$-coalgebra.  The validity of coinduction
  up to $\oplus$ can be derived through the use of distributive laws,
  explained below.
\end{example}

In the context of universal coalgebra, distributive laws are a useful
concept to model interaction between algebra and coalgebra, with
operational semantics as a prominent
example~\cite{TuriP97,Klin11,Bartels04}.  We recall a few basic
notions, and the application to coinduction up-to.

A \emph{distributive law} of a functor $F \colon \C \rightarrow \C$
over a functor $B \colon \C \rightarrow \C$ is a natural
transformation $\lambda \colon FB \Rightarrow BF$.  If $B$ has a final
coalgebra $(Z,\zeta)$, then such a $\lambda$ induces a unique algebra
$\alpha$ making the following commute.
\begin{equation}\label{eq:induced-alg}
\vcenter{
\xymatrix{
	FZ \ar[d]_{\alpha} \ar[r]^-{F\zeta} % chktex 3
		& FBZ \ar[r]^{\lambda_Z} % chktex 3
		&  BFZ \ar[d]^{B\alpha} \\ % chktex 3
	Z \ar[rr]_{\zeta} % chktex 3
		& & BZ
}
}
\end{equation}
We call $\alpha$ the \emph{algebra induced by} $\lambda$ (on the final
coalgebra).
%\todo{talk about final bialgebras? DAM: would say no}

Let $(T,\eta,\mu)$ be a monad. A distributive law of $(T,\eta,\mu)$
over $B$ is a natural transformation $\lambda
\colon TB \Rightarrow BT$ such that $B\eta = \lambda \circ
\eta B$ and $\lambda \circ \mu B = B\mu \circ \lambda T \circ
T\lambda$.

%
%the following diagrams commute:
%$$
%\xymatrix{
% B \ar@{=>}[r]^-{\eta B} \ar@{=>}[dr]_-{B\eta}
% & TB \ar@{=>}[d]^-{\lambda}
% \\
% & BT
%}
%\qquad
%\xymatrix{
% TTB \ar@{=>}[d]_-{\mu B} \ar@{=>}[r]^-{T\lambda}
% & TBT \ar@{=>}[r]^-{\lambda T}
% & BTT \ar@{=>}[d]^-{B\mu}
% \\
% TB \ar@{=>}[rr]_-{\lambda} && BT
% }
%$$
%Bartels~\cite[Theorem 3.8 and Theorem 5.6]{Bartels03}
Bartels~\cite[Theorem 4.2.2, Corollary 4.3.6]{Bartels04} (also see~\cite{Bartels03})
makes the distributive law $\lambda$ explicit in his notion of
$\lambda$-coiteration, and proves (in our terminology) that
coinduction up to the algebra induced by a distributive law is valid,
under certain conditions:
%\footnote{Theorem 5.6 of~\cite{Bartels03} mentions distributive laws
%of a monad over a copointed functor. The case of monad over functor,
%which we consider in Theorem~\ref{thm:valid-bartels}, is explicitly stated and proved, e.g., in~\cite[Corollary 4.3.6]{Bartels04}.}

\begin{theorem}\label{thm:valid-bartels}
  Let $\lambda \colon TB \Rightarrow BT$ be a distributive law between
  functors $B,T \colon \C \rightarrow \C$, let $(Z,\zeta)$ be a final
  $B$-coalgebra and let $\alpha \colon TZ \rightarrow Z$ be the
  algebra induced by $\lambda$. If either
\begin{itemize}
\item $\C$ has countable coproducts, or
\item $T$ is a monad and $\lambda$ a distributive law of that monad
  over $B$,
\end{itemize}
then coinduction up to $\alpha$ is valid.
% For every morphism $f \colon X \rightarrow BTX$ there exists a unique
%morphism $\sol{f} \colon X \rightarrow Z$ such that the following diagram commutes:
%$$
%\xymatrix{
%	X \ar[d]_{f} \ar[rr]^{\sol{f}}
%		& & Z \ar[d]^{\zeta} \\
%	BTX \ar[r]_{BTf}
%		& BTZ \ar[r]_{B\alpha}
%		& BZ
%}
%$$
%where $\alpha$ is the algebra induced by $\lambda$.
\end{theorem}
%\begin{proof}
%  Both implications are instances of the $\lambda$-coiteration scheme
%  of~\cite{Bartels04}, see its Theorem 4.2.2 and Corollary 4.3.6.
%\end{proof}
\begin{example}
  Let $S,B \colon \Set \rightarrow \Set$ be as in the introduction (and Example~\ref{ex:shuffle-coin-up-to}),
  and define $\lambda \colon SB \Rightarrow BS$ by
  \[
      \lambda_X((o_1,t_1),(o_2,t_2)) = (o_1 + o_2,(t_1,t_2)) \,.
  \]
  The algebra induced by $\lambda$ on the final $B$-coalgebra
  $\real^\omega$ is given by pointwise addition
  $\oplus$, as in Equation~\ref{eq:plus1}. Theorem~\ref{thm:valid-bartels} asserts (since $\Set$ has
  countable coproducts) the validity of coinduction up to $\oplus$.
\end{example}

\subsection{Generalisation to morphisms of endofunctors}%
\label{sec:gen-dl}

Below we sometimes consider natural transformations of the more
general form $\lambda \colon FA \Rightarrow BF$, for endofunctors
$A \colon \A \rightarrow \A$, $B \colon \C \rightarrow \C$ and a
functor $F \colon \A \rightarrow \C$.  These form an instance of
\emph{morphisms of endofunctors}: 1-cells in a certain category of
endofunctors in a 2-category~\cite{LenisaPW00}. They appear
in~\cite{CancilaHL03} under the name of \emph{generalized distributive
  laws}.  If $A$ has a final coalgebra $(Z_A, \zeta_A)$, and $B$ a
final coalgebra $(Z_B,\zeta_B)$, then such a $\lambda$ induces a
unique map $\alpha \colon FZ_A \rightarrow Z_B$ such that
\begin{equation}\label{eq:unique-map}
\vcenter{
\xymatrix{
	FZ_A \ar[d]_{\alpha} \ar[r]^-{F\zeta_A} % chktex 3
		& FAZ_A \ar[r]^{\lambda_{Z_A}} % chktex 3
		&  BFZ_A \ar[d]^{B\alpha} \\ % chktex 3
	Z_B \ar[rr]_{\zeta_B} % chktex 3
		& & BZ_B
}
}
\end{equation}
This is, of course, a generalisation of algebras induced by
distributive laws.

\begin{example}\label{ex:even-gen-dl}
  Define $F,A,B \colon \Set \rightarrow \Set$ by $FX = X$,
  $AX = \real \times \real \times X$ and $BX = \real \times X$.  The
  set $\real^\omega$ of streams over $\real$ is a final coalgebra for
  $A$.

  Consider $\lambda \colon FA \Rightarrow BF$ given by
  \[
  \lambda_X(o_1, o_2, t) = (o_1, t) \,.
  \]
  The unique map in~\eqref{eq:unique-map} is given by
  $\even \colon \real^\omega \rightarrow \real^\omega$,
  $\even(\sigma) = (\sigma(0),\sigma(2),\sigma(4),\ldots)$.

  It is straightforward to generalise
  Definition~\ref{def:validity-coind} to this setting, involving
  unique solutions of coalgebras $f \colon X \rightarrow
  AFX$. However, Theorem~\ref{thm:valid-bartels} does not generalise
  accordingly.  For instance, for $F,A,B$ as above,
  $f \colon 1 \rightarrow FA1$ given by $f(\ast) = (0,1,\ast)$, the
  validity of the coinduction up to $\even$ principle would amount to
  a unique solution of
  \[
  x_0 = 0 \quad x_1 = 1 \quad x'' = \even(x)\,,
  \]
  which, however, has multiple solutions.
\end{example}

\section{Properties of the companion}%
\label{sec:companion}

We define the notion of companion, and prove several of its abstract
properties.

\begin{definition}
  Let $B \colon \C \rightarrow \C$ be a functor.  The category
  $\DL{B}$ of distributive laws is defined as follows.  An object is a
  pair $(F, \lambda)$ where $F \colon \C \rightarrow \C$ is a functor
  and $\lambda \colon FB \Rightarrow BF$ is a natural transformation.
  A morphism from $(F, \lambda)$ to $(G,\rho)$ is a natural
  transformation $\kappa \colon F \Rightarrow G$ such that
  $\rho\circ \kappa B = B\kappa\circ\lambda$.
  \[
  \xymatrix{
    FB \ar@{=>}[d]_{\lambda} \ar@{=>}[r]^{\kappa B} % chktex 3
    & GB \ar@{=>}[d]^{\rho} \\ % chktex 3
    BF \ar@{=>}[r]_{B \kappa} % chktex 3
    & BG
  }
  \]
  The \emph{companion} of $B$ is the final object of $\DL{B}$, if it
  exists.
\end{definition}

%In the remainder of this section we restrict our attention to the case $A=B$.
%We write $\DL{B}$ instead of $\DL{A,B}$ and refer to its final object
%as the companion of $B$.
Morphisms in $\DL{B}$ are a special case of \emph{morphisms of
  distributive laws},
see~\cite{PowerW02,Watanabe02,LenisaPW00,KlinN15}.
% \marginpar{possibly we should use more distinct notation for the companion than $T$}
% damien: bof
% Throughout this section we assume an endofunctor $B \colon \C \rightarrow \C$ whose
% companion exists, denoted by $(T,\tau)$.
% We first prove that $(T,\tau)$ is a monad.
We first show that when it exists, the companion is a monad.

\begin{theorem}\label{thm:monad-on-final}
  Let $(T,\tau)$ be the companion of an endofunctor $B$.  There are unique
  $\eta \colon \Id \Rightarrow T$ and $\mu \colon TT \Rightarrow T$
  such that $(T,\eta,\mu)$ is a monad and
  $\tau \colon TB \Rightarrow BT$ is a distributive law of this monad
  over $B$.
\end{theorem}
\begin{proof}
  Define $\eta$ and $\mu$ as the unique morphisms from $\id_B$ and
  $\tau T \circ T\tau$ respectively to the companion:
  \[
  \xymatrix{
    B \ar@{=>}[d]_-{\eta B} \ar@{=}[r] % chktex 3
    & B \ar@{=>}[d]^-{B\eta} \\ % chktex 3
    TB  \ar@{=>}[r]^-{\tau} % chktex 3
    & BT
  }
  \qquad
  \xymatrix{
    TTB \ar@{=>}[d]_-{\mu B} \ar@{=>}[r]^-{T\tau} % chktex 3
    & TBT \ar@{=>}[r]^-{\tau T} % chktex 3
    & BTT \ar@{=>}[d]^-{B\mu} % chktex 3
    \\
    TB \ar@{=>}[rr]^-{\tau} && BT % chktex 3
  }
  \]
  By definition, they satisfy the required axioms for $\tau$ to be a
  distributive law of monad over functor. The proof that
  $(T,\eta,\mu)$ is indeed a monad is routine, using finality of
  $(T,\tau)$; see~\cite{BPR17} for more details.
\end{proof}
Since the companion is a distributive law
of a monad (Theorem~\ref{thm:monad-on-final}), by Theorem~\ref{thm:valid-bartels} we obtain the following.
\begin{corollary}\label{cor:solutions-comp}
  When a functor has a companion,
  coinduction up to the algebra induced by the companion is valid.
  %Let $(Z,\zeta)$ be a final $B$-coalgebra. For every morphism
%  $f \colon X \rightarrow BTX$ there is a unique morphism
%  $\sol{f} \colon X \rightarrow Z$ such that the following commutes:
%  $$
%  \xymatrix{
%    X \ar[d]_{f} \ar[rr]^{\sol{f}}
%    & & Z \ar[d]^{\zeta} \\
%    BTX \ar[r]_-{BT\sol{f}}
%    & BTZ \ar[r]_{B\alpha}
%    & BZ
%  }
%  $$
%  where $\alpha$ is the algebra induced by the distributive law $\tau$
%  of the companion.
\end{corollary}
Instantiated to the complete lattice case, this implies a soundness result (cf.\ Section~\ref{sec:dl}):
any invariant up to the companion (a post-fixpoint of $b\circ t$) is
below the greatest fixpoint ($\nu b$).

%\begin{corollary}
%	There exists a functor $K$ and a natural transformation as in the following diagram:
%        \marginpar{explain}
%	$$
%	\xymatrix{
%		\coalg{BT} \ar[r]^{K} \ar[d]
%			& \coalg{B} \ar[d] \\
%		\C \ar[r]_{\Id}  \ar@{}[ur]|{\rotatebox[origin=c]{45}{$\Rightarrow$}}
%			& \C
%	}
%	$$
%\end{corollary}
%\begin{proof}
%	Apply Theorem~\ref{thm:monad-on-final} and the generalized powerset construction.
%\end{proof}

If the underlying category has an intial object, then one can define the
final coalgebra and the algebra induced by the companion explicitly:
\begin{theorem}%
  \label{thm:T0final}
  Suppose $\C$ has an initial object $0$.
  Let $(T,\tau)$ be the companion of a functor $B \colon \C \rightarrow \C$ and let
  $(T,\eta,\mu)$ be the corresponding monad.
  % given by Theorem~\ref{thm:monad-on-final}.
  %
  The $B$-coalgebra $(T0,\tau_0\circ T!_{B0})$ is final, and the
  algebra induced on it by the companion is given by $\mu_0$.
\end{theorem}

\begin{proof}
  Let $(X,f)$ be a $B$-coalgebra. Write $\hat X$ for the
  constant-to-$X$ functor, and $\hat f$ for the constant-to-$f$
  distributive law of $\hat X$ over $B$. By finality of the companion,
  there exists a unique natural transformation
  $\lambda \colon  \hat X\Rightarrow T$ such that
  $B\lambda \circ \hat f = \tau \circ \lambda B$. One checks easily
  that $\lambda_0$ is the unique coalgebra homomorphism from $(X,f)$
  to $(T0,\tau_0\circ T!_{B0})$.

  To prove that $\mu_0$ is the algebra induced by the companion, it
  suffices to prove that it is a coalgebra morphism of the correct type~\eqref{eq:induced-alg}:
  \begin{align*}
   (\tau_0\circ T!_{B0})\circ \mu_0
   	= \tau_0 \circ \mu_{B0} \circ TT!_{B0}
   	= B\mu_0\circ \tau_{T0}\circ T\tau_0\circ TT!_{B0} = B\mu_0\circ \tau_{T0}\circ T(\tau_0\circ T!_{B0})
  \end{align*}
%  $(\tau_0\circ T!_{B0})\circ \mu_0 = B\mu_0\circ \tau_{T0}\circ
%  T(\tau_0\circ T!_{B0})$,
  which follows from naturality of $\mu$, the fact that $\tau$ is a
  distributive law of the monad $(T,\eta,\mu)$ over $B$, and functoriality.
\end{proof}

More generally, the algebra induced by any distributive law factors
through the algebra $\mu_0$ induced by the companion.

\begin{proposition}\label{prop:alg-on-final}
  Let $(T,\tau)$ be the companion of an endofunctor $B$ and let
  $(T,\eta,\mu)$ be the corresponding monad.
  % given by Theorem~\ref{thm:monad-on-final}.
  %
  Let $\lambda \colon FB \Rightarrow BF$ be a distributive law, and
  $\alpha \colon FT0 \Rightarrow T0$ the algebra on the final
  coalgebra induced by it. Let $\bar{\lambda} \colon F \Rightarrow T$
  be the unique natural transformation induced by finality of the
  companion. Then $\alpha = \mu_0 \circ \bar{\lambda}_{T0}$.
%\begin{enumerate}
% \item Let $\lambda \colon FB \Rightarrow BF$ be a distributive law,
%   and $\alpha \colon FT0 \Rightarrow T0$ the algebra on the final
%   coalgebra induced by it. Let
%   $\bar{\lambda} \colon F \Rightarrow T$ be the unique morphism
%   induced by finality. Then
%   $\alpha = \mu_0 \circ \bar{\lambda}_{T0}$.
% \item If $\lambda \colon FB \Rightarrow BF$ is a distributive
%   law of a monad $(F,\eta^F,\mu^F)$ over the functor $B$, then the
%   unique morphism $\bar{\lambda} \colon F \Rightarrow T$ is a monad
%   morphism from $(F,\eta^F,\mu^F)$ to $(T,\eta,\mu)$.
%\end{enumerate}
\end{proposition}

\begin{proof}
  By Theorem~\ref{thm:T0final},
  $\tau_0 \circ T!_{B0} \colon T0 \rightarrow BT0$ is a final
  $B$-coalgebra. By definition of the algebra induced on the final
  coalgebra by $\lambda$, and uniqueness of morphisms into final
  coalgebras, it suffices to prove that the following diagram
  commutes.
  \[
  \xymatrix{ FT0 \ar[r]^{\bar{\lambda}_{T0}} \ar[d]_{FT!_{B0}} & TT0 % chktex 3
    \ar[d]^{TT!_{B0}} \ar[r]^{\mu_0} % chktex 3
    & T0 \ar[d]^{T!_{B0}} \\ % chktex 3
    FTB0 \ar[d]_{F\tau_0} \ar[r]^{\bar{\lambda}_{TB0}} & TTB0 % chktex 3
    \ar[d]^{T\tau_0} \ar[r]^{\mu_{B0}} % chktex 3
    & TB0 \ar[dd]^{\tau_0} \\ % chktex 3
    FBT0 \ar[d]_{\lambda_{T0}} \ar[r]^{\bar{\lambda}_{BT0}} & TBT0 % chktex 3
    \ar[d]^{\tau_{T0}} % chktex 3
    & \\
    BFT0 \ar[r]_{B\bar{\lambda}_{T0}} & BTT0 \ar[r]_{B\mu_0} & BT0 } % chktex 3
  \]
  Everything commutes, clockwise starting from the top right by
  naturality, the fact that $\tau$ is a
  distributive law of the monad $(T,\eta,\mu)$ over $B$,
  the fact that $\bar{\lambda}$ is a
  morphism from $(F,\lambda)$ to $(T,\tau)$, and twice
  naturality.
\end{proof}

\section{Right Kan extensions and codensity monads}\label{sec:codensity}

The notion of \emph{codensity monad} is a special instance of a right
Kan extension, which plays a central role in the following sections. We
briefly define these notions here;
see, e.g.,~\cite{mac1998categories,nlabrkan,leinster} for a comprehensive study.

Let $F \colon \C \rightarrow \D$, $G \colon \C \rightarrow \E$ be two
functors. Define the category $\K(F,G)$ whose objects are pairs
$(H, \alpha)$ of a functor $H \colon \D \rightarrow \E$ and a
natural transformation $\alpha \colon HF \Rightarrow G$. A morphism
from $(H,\alpha)$ to $(I,\beta)$ is a natural transformation
$\kappa \colon H \Rightarrow I$ such that
$\beta \circ \kappa F = \alpha$.
\[
\xymatrix{ HF \ar@{=>}[rr]^{\kappa F} \ar@{=>}[dr]_{\alpha} % chktex 3
  & & IF \ar@{=>}[dl]^{\beta} \\ % chktex 3
  & G & }
\]
The \emph{right Kan extension} of $G$ along $F$ is a final object
$(\Ran{F}{G},\epsilon)$ in $\K(F,G)$; the natural transformation
$\epsilon\colon \Ran{F}{G} F\Rightarrow G$ is called its
\emph{counit}.  A functor $K \colon \E \rightarrow \F$ is said to
\emph{preserve} $\Ran{F}{G}$ if $K \circ \Ran{F}{G}$ is a right Kan
extension of $KG$ along $F$, with counit
$K\epsilon \colon K\Ran{F}{G}F \Rightarrow KG$.

\smallskip

The codensity monad is a special case, with $F=G$.
Explicitly, the \emph{codensity monad} of a functor
$F \colon \C \rightarrow \D$ consists of a functor
$\cod{F} \colon \D \rightarrow \D$ and a natural transformation
$\epsilon \colon \cod FF\Rightarrow F$ such that for every functor
$H \colon \D \rightarrow \D$ and natural transformation
$\alpha \colon HF \Rightarrow F$ there is a unique
$\hat{\alpha} \colon H \Rightarrow \cod{F}$ such that
$\epsilon \circ \hat{\alpha} F = \alpha$.
\[
\xymatrix{ HF \ar@{=>}[rr]^{\hat{\alpha} F} \ar@{=>}[dr]_{\alpha} % chktex 3
  & & \cod{F}F \ar@{=>}[dl]^{\epsilon} \\ % chktex 3
  & F & }
\]
As the name suggests, $\cod{F}$ is a monad: the unit $\eta$ and the
multiplication $\mu$ are the unique natural transformations such that
the following diagrams commute.
\begin{equation}\label{eq:def-cod-monad}
\begin{gathered}
\xymatrix{
	F \ar@{=>}[r]^{\eta F} \ar@{=}[dr] % chktex 3
		& \cod{F} F \ar@{=>}[d]^{\epsilon} \\ % chktex 3
		& F
}
\qquad
\xymatrix{
	\cod{F} \cod{F} F \ar@{=>}[r]^{\cod{F} \epsilon} \ar@{=>}[d]_{\mu F} % chktex 3
		& \cod{F} F \ar@{=>}[d]^{\epsilon} \\ % chktex 3
	\cod{F} F \ar@{=>}[r]_{\epsilon} % chktex 3
		& F
}
\end{gathered}
\end{equation}
In the sequel we abbreviate the category $\K(F,F)$ as $\K(F)$.
%$\epsilon \circ \eta F = \id$ and
%$\epsilon \circ \mu F = \epsilon \circ \cod{F} \epsilon$.  In the
%sequel we will abbreviate the category $\K(F,F)$ as $\K(F)$.

\smallskip

Right Kan extensions (and in particular codensity monads) can be computed pointwise as a limit, if
sufficient limits exist.  For an object $X$ in $\D$, denote by
$\Delta_X \colon \C \rightarrow \D$ the functor that maps every object
to $X$.  By $\Delta_X / F$ we denote the comma category, where an
object is a pair $(Y,f)$ consisting of an object $Y$ in $\C$ and an
arrow $f \colon X \rightarrow FY$ in $\D$, and an arrow from $(Y,f)$
to $(Z,g)$ is a map $h \colon Y \rightarrow Z$ in $\C$ such that
$Fh \circ f = g$.
%to $(Z,g)$ is a map $h \colon Y \rightarrow Z$ in $\C$ making the
%following triangle commute:
%$$
%\xymatrix{
%& X \ar[dl]_{f} \ar[dr]^g & \\
%FY \ar[rr]_{Fh} & & FZ
%}
%$$
There is a forgetful functor $(\Delta_X / F) \rightarrow \C$, which
remains unnamed below.

\begin{lemma}\label{lm:pw-codensity}
  Let $F \colon \C \rightarrow \D$, $G \colon \C \rightarrow \E$ be
  functors.  If, for every object $X$ in $\D$, the limit
  $ \lim\left((\Delta_X / F) \rightarrow \C \xrightarrow{G} \E\right)
  $
  exists, then the right Kan extension $\Ran{F}{G}$ exists, and is
  given on an object $X$ by the corresponding limit.
%  Let $F \colon \C \rightarrow \D$, $G \colon \C \rightarrow \E$ be
%  functors.  If, for every object $X$ in $\D$, the following limit
%  exists:
%  $$
%  \lim\left((\Delta_X / F) \rightarrow \C \xrightarrow{G}
%    \D\right)
%  $$
%  then the right Kan extension $\Ran{F}{G}$ exists, and is given on an
%  object $X$ by the above limit.
\end{lemma}
%
%
%The codensity monad of a functor $F$ is the right Kan extension of $F$
%along itself. Hence, Lemma~\ref{lm:pw-codensity} gives us a way of
%computing the codensity monad.

The hypotheses of Lemma~\ref{lm:pw-codensity} are met in particular if
$\C$ is essentially small (equivalent to a category with a set
of objects and a set of arrows), $\D$ is locally small and $\E$ is complete. The latter
conditions hold for $\D = \E =\Set$. In that case, we have the following
concrete presentation; see, e.g.,~\cite[Section 2.5]{cordier2008shape}
for a proof.
% It can be formulated at the general level of Kan
% extensions, but we only need it for the codensity monad.
% \marginpar{I'm not sure about the status of the lemma; my feeling is
%   it should be very standard and old; the reference is just one i
%   bumped into, they give a proof but don't cite anything}
\begin{lemma}\label{lm:codensity-set}
  Let $F,G \colon \C \rightarrow \Set$ be functors.
  Suppose that, for each set $X$, the collection $\{\alpha \colon {(F-)}^X \Rightarrow G\}$
  is a set (rather than a proper class).
  Then the right Kan extension $\Ran{F}{G}$ is given by
 \[\Ran{F}{G}(X) = \{\alpha \colon {(F-)}^X \Rightarrow G\}\]
 for each $X$.
	For $h \colon X \rightarrow Y$,
  	${(\Ran{F}{G}(h)(\alpha))}_A \colon  {(FA)}^Y \rightarrow GA$
  	is given by $f \mapsto \alpha_A(f \circ h)$.
%  Let $F,G \colon \C \rightarrow \Set$ be functors, where $\C$ is
%  essentially small.  The right Kan extension $\Ran{F}{G}$ is given by
% $$\Ran{F}{G}(X) = \{\alpha \colon (F-)^X \Rightarrow G\}$$
% and, for $h \colon X \rightarrow Y$,
%  	$(\Ran{F}{G}(h)(\alpha))_A  \colon  (FA)^Y \rightarrow GA$
%  	is given by $f \mapsto \alpha_A(f \circ h)$.

  	The natural transformation
  	  $\epsilon \colon \Ran{F}{G} F \Rightarrow G$ is given by
  $\epsilon_X(\alpha \colon F^{FX} \Rightarrow G) =
  \alpha_X(\id_{FX})$.
  Finally, given $H \colon \Set \rightarrow \Set$ and $\beta \colon HF \Rightarrow G$,
  the induced $\hat{\beta} \colon H \Rightarrow \Ran{F}{G}$ is
  given by ${(\hat{\beta}_X(S))}_A \colon {(FA)}^X \rightarrow GA$, $f \mapsto \beta_A \circ Hf(S)$.
%
%  codensity monad $\cod{F}$ is given by
%  $\cod{F}(X) = \{\alpha \colon (F-)^X \Rightarrow F\}$ and, for
%  $h \colon X \rightarrow Y$,
%  $(\cod{F}(h)(\alpha))_A \colon (FA)^Y \rightarrow FA$ is given by
%  $f\mapsto \alpha_A(f \circ h)$.  The natural transformation
%  $\epsilon \colon \cod{F} F \Rightarrow F$ is given by
%  $\epsilon_X(\alpha \colon F^{FX} \Rightarrow F) =
%  \alpha_X(\id_{FX})$.
\end{lemma}
%\todo{I changed a bit since $\C$ is not essentially small in our
%  examples; but not so nice formulation. DAM: indeed not so nice, but
%  no choice, right?}
%\todo{check carefully this generalisation to right kan extensions}
%shape theory Categorical Methods of Approximation "be the kan extension of"

\section{Constructing the companion via right Kan extensions}

It is standard in the theory of coalgebras to compute the final
coalgebra of a functor $B$ as a limit of the final sequence $\fin{B}$,
see Section~\ref{sec:prelim}.
In this section,
we show how the companion of a functor arises as the codensity monad of its final sequence.
%In this section,
%we show how the companion of a pair of functors $(A,B)$ arises
%as a right Kan extension of the final sequence of $B$ along the final sequence
%of $A$. In particular, the companion of a functor $B$ arises as
%the codensity monad of its final sequence.

We first adapt the definition of companion to more general natural transformations
of the form $FA \Rightarrow BF$,
fixing \emph{two} functors; such natural transformations were discussed in Section~\ref{sec:gen-dl}.
This generalisation is useful in the next sections, in
the setting of causal functions. Moreover, the
construction of the companion given in this section
can be presented naturally at this level.
\begin{definition}
    Let $A \colon \A \rightarrow \A$ and $B \colon \C \rightarrow \C$  be functors.  The category
    $\DL{A,B}$ is defined as follows.  An object is
    a pair $(F, \lambda)$ where $F \colon \A \rightarrow \C$ is a
    functor and $\lambda \colon FA \Rightarrow BF$ is a natural
    transformation.  A morphism from $(F, \lambda)$ to $(G,\rho)$ is a
    natural transformation $\kappa \colon F \Rightarrow G$ such that
    $\rho\circ \kappa A = B\kappa\circ\lambda$.
    The \emph{companion} of $(A,B)$ is the final object $(T,\tau)$ of $\DL{A,B}$, if it
    exists.
\end{definition}

Recall, from Section~\ref{sec:prelim}, that the final sequence of an
endofunctor $A \colon \A \rightarrow \A$ in a complete category $\A$
can be presented as a functor
$\fin{A} \colon \Ord^{\op} \rightarrow \A$. Given another functor
$B \colon \C \rightarrow \C$, consider the right Kan extension
$\Ran{\fin{A}}{\fin{B}}$ of the final sequence of $\fin{B}$ along the
final sequence of $\fin{A}$.  By definition, this is final in the
category of natural transformations of the form
$\alpha \colon F \fin{A} \Rightarrow \fin{B}$.
The main result of this section is that, under certain conditions, the right Kan extension $\Ran{\fin{A}}{\fin{B}}$
is (the underlying functor of) the companion of $(A,B)$, i.e., the final object in
the category of distributive laws of the form $\lambda \colon FA \Rightarrow BF$ (Theorem~\ref{thm:companion-codensity}).
The following lemma is a first step: it associates, to every such distributive law $\lambda \colon FA \Rightarrow BF$,
a natural transformation of the form $\alpha \colon F \fin{A} \Rightarrow \fin{B}$.

\begin{lemma}\label{lm:dl-alpha}
  For every $\lambda \colon FA \Rightarrow BF$ there exists a unique
  $\alpha \colon F\fin{A} \Rightarrow \fin{B}$ such that for all
  $i \in \Ord$: $\alpha_{i+1} = B\alpha_i \circ \lambda_{A_i}$.
  This construction extends to a functor from $\DL{A,B}$ to $\K(\fin{A},\fin{B})$.

  Moreover, if $A_{k+1,k}$ and $B_{k+1,k}$ are isomorphisms for some $k$, then
  $\alpha_k \colon FA_k \rightarrow B_k$ is the unique map induced by $\lambda$ (as in~\eqref{eq:unique-map}).
%  $$
%	\xymatrix{
%		FA_k \ar[r]^{\alpha_k}\ar[d]_{FA_{k+1,k}^{-1}}
%			& B_k \ar[dd]^{B_{k+1,k}^{-1}} \\
%		FAA_k \ar[d]_{\lambda_{A_k}}
%			& \\
%		BFA_k \ar[r]_{B\alpha_k}
%			& BB_k
%	}
%  $$
  In particular, if $A=B$ then $\alpha_k$ is the algebra induced by $\lambda$ on the final coalgebra.
\end{lemma}

\begin{proof}
	This natural transformation is completely determined by the successor case given in the definition;
	on a limit ordinal $j$, $B_j$ is a limit, and naturality requires it to be defined as the unique arrow
	$\alpha_j \colon FA_j \rightarrow B_j$ such that
	\[
	\xymatrix{
		FA_j \ar[d]_{FA_{j,i}} \ar@{-->}[r]^{\alpha_j} % chktex 3
			& B_j \ar[d]^{B_{j,i}} \\ % chktex 3
		FA_i \ar[r]_{\alpha_i} % chktex 3
			& B_i
	}
	\]
	commutes, for all $i < j$.

	For naturality, we have to prove that the relevant square (as above) commutes for all $i,j$ with $i \leq j$.
	For $i=j$, this follows since $A_{j,j} = \id_{A_j}$ and $B_{j,j} = \id_{B_j}$ by definition
	of the final sequence. We prove that the square commutes
	for any $i,j$ with $i < j$, by induction on $j$. The case that $j$ is a limit ordinal follows immediately from
	the definition of $\alpha_j$, without using the induction hypothesis.

	Now suppose that, for any $i$ with $i<j$, the square commutes for $i,j$. We need to
	prove that it commutes for all $i < j+1$. First observe that if $i < j$,
	then the square also commutes for $i+1<j+1$:
	\[
	\xymatrix@C=1.5cm{
		FA_{j+1} = FAA_j \ar[r]_-{\lambda_{A_j}}\ar[d]_{FA_{j+1,i+1} = FAA_{j,i}} \ar@/^15pt/[rr]^{\alpha_{j+1}} % chktex 3 chktex 25
			& BFA_j \ar[r]_{B\alpha_j} \ar[d]^{BFA_{j,i}} % chktex 3
			& BB_j = B_{j+1} \ar[d]^{B_{j+1,i+1} = BB_{j,i}} \\ % chktex 3
		FA_{i+1} = FAA_i \ar[r]^-{\lambda_{A_i}} \ar@/_15pt/[rr]_{\alpha_{i+1}} % chktex 3 chktex 25
			& BFA_i \ar[r]^{B\alpha_i} % chktex 3
			& BB_i = B_{i+1}
	}
	\]
	by naturality (left square), assumption (right square) and definition of $\alpha$ on successor ordinals (crescents). Hence, the square commutes for any successor ordinal
	$i+1$ strictly below $j+1$.

	For $i$ a limit ordinal, consider the following diagram:
\[
\xymatrix{
	FA_j \ar[r]^{FA_{j,i}} \ar[d]_{\alpha_j} % chktex 3
		& FA_i \ar[r]^{FA_{i,l}} \ar[d]^{\alpha_i} % chktex 3
		& FA_l \ar[d] \ar[d]^{\alpha_l}\\ % chktex 3
	B_j \ar[r]_{B_{j,i}} % chktex 3
		& B_i \ar[r]_{B_{i,l}} % chktex 3
		& B_l
}
\]
For all $l < i$, the outer rectangle commutes by the induction hypothesis, and the right square by definition of $\alpha_i$
on the limit ordinal $i$. Since $B_i$ is a limit with projections $B_{i,l}$ for $l \leq i$, it follows
that the square on the left commutes, as desired.

To show that the construction extends to a functor,
let $\kappa \colon F \Rightarrow G$ be a morphism in $\DL{A,B}$
from some $(F,\lambda)$ to $(G,\rho)$. The natural transformations $(F,\lambda)$ and $(G,\rho)$ respectively
yield unique $\alpha^\lambda \colon F\fin{A} \Rightarrow \fin{B}$ and $\alpha^\rho \colon G \fin{A} \Rightarrow \fin{B}$ in
the above way. We need to prove that $\kappa$ is a morphism in $\K(\fin{A},\fin{B})$, i.e., that
\[
\xymatrix{
	F\fin{A} \ar@{=>}[rr]^{\kappa \fin{A}} \ar@{=>}[dr]_{\alpha^\lambda} % chktex 3
		& & G \fin{A} \ar@{=>}[dl]^{\alpha^\rho} \\ % chktex 3
		& \fin{B} &
}
\]
commutes. This has a straightforward proof by induction; for the successor case one uses that $\kappa$ is a morphism in $\DL{A,B}$,
and for limit ordinals $j$ the universal property of the limit $B_j$, the induction hypothesis and the definition of $\alpha^\lambda_j$
and $\alpha^\rho_j$.

For the final point in the statement: if $B_{k+1,k} \colon B_{k+1} \rightarrow B_{k}$ is an isomorphism, then $B_{k+1,k}^{-1} \colon
B_k \rightarrow B(B_{k+1})$ is a final $B$-coalgebra, and similarly for $A$ and $A_k$. Hence,
to show that $\alpha_k$ is the unique map induced by $\lambda$ as in~\eqref{eq:unique-map},
it suffices to show that the following diagram commutes:
\[
\xymatrix{
	FA_k \ar[r]^{\alpha_k}\ar[d]_{FA_{k+1,k}^{-1}} % chktex 3
		& B_k \ar[dd]^{B_{k+1,k}^{-1}} \\ % chktex 3
	FAA_k \ar[d]_{\lambda_{A_k}} \ar[dr]^{\alpha_{k+1}} % chktex 3
		& \\
	BFA_k \ar[r]_{B\alpha_k} % chktex 3
		& BB_k
}
\]
The triangle commutes by definition of $\alpha$, and the shape above it by naturality and the fact
that $A_{k+1,k}$ and $B_{k+1,k}$ are isomorphisms.
\end{proof}
The natural transformation $\alpha$ arising from $\lambda$ by the above lemma yields a natural transformation $\hat{\alpha} \colon F \Rightarrow \kfab$
due to the universal property of the right Kan extension. This will be shown to be the unique morphism
in $\DL{A,B}$, turning $\kfab$ into the companion. However this requires
a natural transformation $\kfab A \Rightarrow B \kfab$.
For its existence, we assume that $B$ preserves the right Kan extension $\kfab$.
This condition, as well as the concrete form of the
companion computed in this manner, becomes clearer when we instantiate
this result to the case of lattices
(Section~\ref{sec:cod-lattice}) and to $\Set$
(Section~\ref{sec:causal}).

\begin{theorem}\label{thm:companion-codensity}
	Let $A \colon \A \rightarrow \A$ and $B \colon \C \rightarrow \C$
	be endofunctors. %, with final sequences $\fin{A}$ and $\fin{B}$ respectively.
	Suppose
  the right Kan extension $\kfab$ exists and
  $B$ preserves it. Then there is a
  natural transformation $\tau \colon \kfab A \Rightarrow B \kfab$
  such that $(\kfab,\tau)$ is the companion of $(A, B)$.
\end{theorem}
%\begin{proof}[Outline]
%  The preservation assumption means that $(B\cf{B}, B\epsilon)$ is a
%  right Kan extension of $B\fin{B}$ along $\fin{B}$. The natural
%  transformation $\tau$ is defined, using the universal property of
%  $B\epsilon$, as the unique
%  $\tau \colon \cf{B} B \Rightarrow B\cf{B}$ such that
%  $B\epsilon_i \circ \tau_{B_i} = \epsilon_{i+1} \colon \cf{B} B B_i
%  \Rightarrow BB_i$
%  for all $i$.  See the appendix for a full proof.
%\end{proof}

\begin{proof}
%The preservation condition can be phrased as commutativity of the following diagram:
%
%where $\iota$ is an isomorphism, and $\epsilon$ and $\epsilon'$ are the natural transformations
%associated to $\cf{B} = \Ran{\fin{B}}{\fin{B}}$ and $\Ran{B \circ \fin{B}}{\fin{B}}$ respectively.
By assumption, $(B\kfab, B\epsilon)$ is a right Kan extension of $B\fin{B}$
along $\fin{A}$. This means that for all $\alpha \colon H\fin{A} \Rightarrow B\fin{B}$,
there exists a unique $\hat{\alpha} \colon H \Rightarrow B\kfab$
such that $\alpha = B\epsilon \circ \hat{\alpha} {\fin{A}}$.
We use this universal property to define the natural transformation $\tau$, choosing
$H = \kfab A$.

To this end, consider the functor $S \colon \Ord^\op \rightarrow \Ord^\op$ defined by $S(i) = i+1$.
For any $F \colon \C \rightarrow \C$, we have
\begin{equation}\label{eq:b-s}
%\xymatrix{
%	\Ord^\op \ar[r]^{\fin{F}} \ar[d]_{S}
%		& \C \ar[d]^{F} \\
%	\Ord^\op \ar[r]_{\fin{F}}
%		& \C
%}
\fin{F} S = F \fin{F} \,
\end{equation}
simply expressing that $F_{i+1} = FF_i$ and $F_{j+1,i+1} = FF_{j,i}$ for all $i\leq j$, which both hold by definition of the final sequence.
As a consequence, there is the natural transformation on the top row of the diagram below:
\begin{equation}\label{eq:tau-in-proof}
\begin{gathered}
\xymatrix{
	\kfab A \fin{A} \ar@{=}[r] \ar@{=>}[d]_{\tau \fin{A}} % chktex 3
		& \kfab \fin{A} S \ar@{=>}[r]^-{\epsilon S} % chktex 3
		& \fin{B} S  \ar@{=}[r]
		& B\fin{B} \\
	B\kfab \fin{A} \ar@{=>}[urrr]_{B\epsilon} % chktex 3
		& & &
}
\end{gathered}
\end{equation}
By the universal property of $(B \kfab, B\epsilon)$ we obtain
$\tau \colon \kfab B \Rightarrow B \kfab$
as the unique natural transformation making the above diagram~\eqref{eq:tau-in-proof} commute.

We now show that
$(\kfab, \tau)$ is the companion of $(A,B)$, i.e., that it is final in the category $\DL{A,B}$.
Let $\lambda \colon FA \Rightarrow BF$ be a natural transformation.
We need to prove that there exists a unique
$\hat{\alpha} \colon F \Rightarrow \kfab$
making the following diagram commute:
\begin{equation}\label{eq:psi-morph}
\begin{gathered}
\xymatrix{
	FA \ar@{=>}[r]^-{\hat{\alpha} A } \ar@{=>}[d]_{\lambda} % chktex 3
		& \kfab A \ar@{=>}[d]^{\tau}\\ % chktex 3
	BF \ar@{=>}[r]_-{B\hat{\alpha}} % chktex 3
		& B \kfab
}
\end{gathered}
\end{equation}
First, observe that for every natural transformation of the form
$\alpha \colon F\fin{A} \Rightarrow \fin{B}$,
there is a unique $\hat{\alpha} \colon F \Rightarrow \kfab$
such that $\epsilon \circ \hat{\alpha}{\fin{A}} = \alpha$,
by the universal property of $\epsilon$.
We prove that $\hat{\alpha}$ satisfies~\eqref{eq:psi-morph}
if and only if $\alpha$ makes the following diagram commute:
\begin{equation}\label{eq:lambda-alpha}
\begin{gathered}
\xymatrix{
	FA\fin{A} \ar@{=>}[rr]^{\lambda \fin{A}} \ar@{=}[d] % chktex 3
		&
		& BF\fin{A} \ar@{=>}[d]^{B\alpha} \\ % chktex 3
	F\fin{A} S \ar@{=>}[r]_{\alpha S} % chktex 3
		& \fin{B} S \ar@{=}[r]
		& B\fin{B}
}
\end{gathered}
\end{equation}
By Lemma~\ref{lm:dl-alpha}, $\lambda$ induces a unique $\alpha$ making the above diagram commute.
Hence, it then follows that $\hat{\alpha}$ is the unique morphism to $\tau$.

By the universal property of $B\epsilon$,~\eqref{eq:psi-morph} commutes
if and only if the following equation holds:
\begin{equation}\label{eq:equation-bar-psi}
	B\epsilon \circ \tau {\fin{A}} \circ \hat{\alpha} {A\fin{A}}
	= B\epsilon \circ B\hat{\alpha} {\fin{A}} \circ \lambda {\fin{A}}
\end{equation}
Hence, it suffices to prove that~\eqref{eq:lambda-alpha} is equivalent to~\eqref{eq:equation-bar-psi}.

Consider the following diagram:
\[
\xymatrix@R=0.5cm@C=0.5cm{
	FA\fin{A} \ar@{=>}[rrrr]^{\hat{\alpha}A\fin{A}} \ar@{=>}[dddd]_{\lambda \fin{A}} \ar@{=}[dr] % chktex 3
		& & &
		& \kfab A \fin{A} \ar@{=>}[dddd]^{\tau \fin{A}} \ar@{=}[dl] \\ % chktex 3
		& F\fin{A} S \ar@{=>}[rr]^{\hat{\alpha}\fin{A}S} \ar@{=>}[dr]_{\alpha S} % chktex 3
		& & \kfab \fin{A} S \ar@{=>}[dl]^{\epsilon S} % chktex 3
		& \\
		& & \fin{B} S \ar@{=}[d]
		& & \\
		& & B \fin{B}
		& & \\
	B F \fin{A} \ar@{=>}[rrrr]_{B\hat{\alpha}\fin{A}} \ar@{=>}[urr]^{B\alpha} % chktex 3
		& & &
		& B\kfab\fin{A} \ar@{=>}[ull]_{B\epsilon} % chktex 3
}
\]
The two triangles commute by definition of $\hat{\alpha}$, the upper trapezoid
by the equality $A\fin{A} = \fin{A} S$, the right trapezoid by definition of $\tau$.
The left trapezoid is~\eqref{eq:lambda-alpha}. The equivalence of~\eqref{eq:lambda-alpha}
and~\eqref{eq:equation-bar-psi} follows from a straightforward diagram chase.
\end{proof}

In case $A=B$, the right Kan extension $\kfab$ is the codensity monad $\cf{B}$. It turns out that the unit and multiplication
of this codensity monad coincide with the unique monad structure induced on the companion by
Theorem~\ref{thm:monad-on-final}. This follows from uniqueness of such a monad structure turning $\tau$
into a distributive law, together with the following theorem.

\begin{theorem}\label{thm:tau-dl}
  Let $B \colon \C \rightarrow \C$ be an endofunctor such that the
  right Kan extension $\cf B$ exists and $B$ preserves it.  We have
  that $\tau$ is a distributive law of the codensity monad
  $(\cf{B},\eta,\mu)$ over $B$.
\end{theorem}

\begin{proof}
%We prove that $\tau$ is a distributive law of the codensity monad $(\cf{B},\eta,\mu)$
%over $B$.
%To this end, recall that $\eta$ and $\mu$ are defined as the unique natural transformations
%making the following diagrams commute:
%\begin{equation}\label{eq:def-cod-monad}
%\begin{gathered}
%\xymatrix{
%	\fin{B} \ar@{=>}[r]^{\eta \fin{B}} \ar@{=}[dr]
%		& \cf{B} \fin{B} \ar@{=>}[d]^{\epsilon} \\
%		& \fin{B}
%}
%\qquad
%\xymatrix{
%	\cf{B} \cf{B} \fin{B} \ar@{=>}[r]^{\cf{B} \epsilon} \ar@{=>}[d]_{\mu \fin{B}}
%		& \cf{B} \fin{B} \ar@{=>}[d]^{\epsilon} \\
%	\cf{B} \bar{B} \ar@{=>}[r]_{\epsilon}
%		& \fin{B}
%}
%\end{gathered}
%\end{equation}
For the unit axiom, we need to prove $\tau \circ \eta B = B \eta$,
%$$
%\xymatrix{
%	B \ar@{=>}[r]^{\eta B} \ar@{=>}[dr]_{B\eta}
%		& \cf{B} B \ar@{=>}[d]^{\tau} \\
%		& B\cf{B}
%}
%$$
which we do by showing that $B\epsilon \circ \tau {\fin{B}} \circ \eta {B\bar{B}} = B\epsilon \circ B\eta {\bar{B}}$;
the desired equality then follows by the universal property of the right Kan extension $(B\cf{B},B\epsilon)$.
The equality follows from commutativity of:
\[
\xymatrix{
	& & \cf{B} B \bar{B} \ar@{=}[d] \ar@{=>}[ddr]^{\tau \bar{B}} & \\ % chktex 3
  B \bar{B} \ar@{=}[r] \ar@{=>}[urr]^{\eta B \fin{B}} \ar@{=>}[dd]_{B\eta\bar{B}} \ar@{=}[drdr] % chktex 3
	& \bar{B} S \ar@{=}[dr]\ar@{=>}[r]^{\eta \bar{B} S} % chktex 3
	& \cf{B} \bar{B} S \ar@{=>}[d]^{\epsilon S} & \\ % chktex 3
	& & \bar{B} S \ar@{=}[d] & B \cf{B} \bar{B} \ar@{=>}[dl]^{B\epsilon}\\ % chktex 3
  B \cf{B} \bar{B} \ar@{=>}[rr]_{B\epsilon} % chktex 3
  	& & B\bar{B}
}
\]
The two triangles within the big square commute by definition of $\eta$~\eqref{eq:def-cod-monad},
the upper left triangle and the trapezoid in the square since $\bar{B} S = B \bar{B}$ (see~\eqref{eq:b-s}),
and the right triangle by definition of $\tau$ (see~\eqref{eq:tau-in-proof}).

For the other axiom, we need to prove $\tau \circ \mu B = B \mu \circ \tau \cf{B} \circ \cf{B} \tau$
%$$
%\xymatrix{
%	\cf{B} \cf{B} B \ar@{=>}[r]^{\mu B} \ar@{=>}[d]_{\cf{B} \tau}
%		& \cf{B} B \ar@{=>}[dd]^{\tau} \\
%	\cf{B} B \cf{B} \ar@{=>}[d]_{\tau \cf{B}}
%		& \\
%	B\cf{B} \cf{B} \ar@{=>}[r]_{B\mu}
%		& B \cf{B}
%}
%$$
which, in a similar manner as above for the unit, follows from the universal property of $B\epsilon$
and commutativity of the following diagram.
\[
\xymatrix{
	\cf{B}\cf{B} B  \fin{B} \ar@{=>}[rr]^{\mu B\fin{B}} \ar@{=}[dr] \ar@{=>}[ddd]_{\cf{B} \tau \fin{B}} % chktex 3
		& & \cf{B} B \fin{B} \ar@{=}[d] \ar@{=>}[r]^{\tau \fin{B}} % chktex 3
		& B \cf{B} \fin{B} \ar@{=>}[dddd]^{B\epsilon} \\ % chktex 3
		& \cf{B}\cf{B}\fin{B} S   \ar@{=>}[r]^{\mu {\bar{B}S}} \ar@{=>}[d]_{\cf{B}\epsilon S} % chktex 3
		& \cf{B} \bar{B} S \ar@{=>}[d]^{\epsilon S} % chktex 3
		& \\
		& \cf{B} \fin{B} S \ar@{=>}[r]_{\epsilon S} \ar@{=}[d] % chktex 3
		&  \fin{B} S \ar@{=}[ddr]
		& \\
	\cf{B}B\cf{B}\fin{B} \ar@{=>}[dd]_{\tau {\cf{B}\fin{B}}} \ar@{=>}[r]^{\cf{B} B \epsilon} % chktex 3
		& \cf{B} B \fin{B} \ar@{=>}[d]_{\tau \fin{B}} % chktex 3
		&
		& \\
		& B\cf{B} \fin{B} \ar@{=>}[rr]^{B\epsilon} % chktex 3
		&
		& B\fin{B}\\
	B\cf{B}\cf{B} \bar{B} \ar@{=>}[ur]^{B\cf{B}\epsilon} \ar@{=>}[rrr]_{B\mu \fin{B}} % chktex 3
	 & & & B \cf{B} \fin{B} \ar@{=>}[u]_{B\epsilon} % chktex 3
}
\]
The square in the middle commutes by definition of $\mu$ (see~\eqref{eq:def-cod-monad}).
The rest commutes, clockwise starting from the north, by
the equality $\bar{B} S = B \bar{B}$ (see~\eqref{eq:b-s}), twice definition of $\tau$
(see~\eqref{eq:tau-in-proof}), definition of $\mu$ (the south), naturality of $\tau$
and again definition of $\tau$.
\end{proof}

%Every distributive law induces an algebra on the final coalgebra (Section~\ref{sec:dl}).
The following result characterises the algebra induced on the final
coalgebra by the distributive law of the companion, in terms of the
counit $\epsilon$ of the codensity monad of $\fin{B}$. This plays an
important role for the case $\C=\Set$ (Section~\ref{sec:comp-set}).

\begin{lemma}\label{lm:alg-codensity}
  Let $A$ and $B$ be endofunctors such that the right Kan extension
  $\kfab$ exists and $B$ preserves it.
  % Let $(\kfab, \epsilon)$
  % be the right Kan extension, forming the companion with the natural
  % transformation $\tau$.
  If $A_{k+1,k}$ and $B_{k+1,k}$ are both isomorphisms for some $k$,
  then $\epsilon_k \colon \kfab A_k \rightarrow B_k$ is the unique map
  induced by $\tau$ as in the diagram below, where $F$ is a shorthand
  for $\kfab$.
 \[
	\xymatrix{
		FA_k \ar[r]^{\epsilon_k}\ar[d]_{FA_{k+1,k}^{-1}} % chktex 3
			& B_k \ar[dd]^{B_{k+1,k}^{-1}} \\ % chktex 3
		FAA_k \ar[d]_{\tau_{A_k}} % chktex 3
			& \\
		BFA_k \ar[r]_{B\epsilon_k} % chktex 3
			& BB_k
	}
 \]
\end{lemma}

\begin{proof}
  By definition of $\tau$ in the proof of
  Theorem~\ref{thm:companion-codensity}, we have
  $B\epsilon_i \circ \tau_{A_i} = \epsilon_{i+1}$ for all $i$.
  The result follows by Lemma~\ref{lm:dl-alpha}.
  %Hence,
%  the triangle commutes:
%    $$
%	\xymatrix{
%		FA_k \ar[r]^{\alpha_k}\ar[d]_{FA_{k+1,k}^{-1}}
%			& B_k \ar[dd]^{B_{k+1,k}^{-1}} \\
%		FAA_k \ar[d]_{\tau_{A_k}} \ar[dr]^{\epsilon_{k+1}}
%			& \\
%		BFA_k \ar[r]_{B\epsilon_k}
%			& BB_k
%	}
%  $$
%  The other inner shape commutes by naturality of $\alpha$ and that both $A_{k+1,k}$ and $B_{k+1,k}$ are isomorphisms.
\end{proof}

\begin{corollary}
  Let $B \colon \C \rightarrow \C$ be an endofunctor such that the
  right Kan extension $\cf B$ exists and $B$ preserves it.  Let
  $(\cf{B},\epsilon)$ be the codensity monad of $\fin{B}$, with
  distributive law $\tau$ and monad structure
  $(\cf{B},\eta,\mu)$. Further, suppose $\C$ has an initial object
  $0$. If $B_{k+1,k}$ is an isomorphism for some $k$, then $\mu_0$ is
  isomorphic (as a $\cf{B}$-algebra) to $\epsilon_k$, i.e., there is
  an isomorphism $\iota \colon \cf{B}0 \rightarrow B_k$ (in $\C$)
  making the following diagram commute.
  \begin{equation}\label{eq:mu-eps-iso}
  \vcenter{
  \xymatrix{
  	\cf{B} \cf{B} 0 \ar[r]^{\cf{B} \iota} \ar[d]_{\mu_0} % chktex 3
  		& \cf{B} B_k \ar[d]^{\epsilon_k} \\ % chktex 3
  	\cf{B} 0 \ar[r]_{\iota}^{\cong} % chktex 3
  		& B_k
  }
  }
  \end{equation}
\end{corollary}
\begin{proof}
  By Theorem~\ref{thm:tau-dl}, $\tau$ is a distributive law of the codensity
  monad $(\cf{B},\eta,\mu)$ over $B$.  Since $\cf{B}$ is the companion
  of $B$ (Theorem~\ref{thm:companion-codensity}), $\eta$ and $\mu$ coincide with the natural transformations
  in Theorem~\ref{thm:monad-on-final}.  By Theorem~\ref{thm:T0final},
  $\cf{B}0$ is the carrier of a final coalgebra, and $\mu_0$ is the algebra induced
  on it by $\tau$.

  Since $B_{k+1,k}$ is an isomorphism, $B_k$ is a final coalgebra, and hence, since $\cf{B}0$ is also final,
  there is an isomorphism (of coalgebras) $\iota \colon \cf{B}0 \rightarrow B_k$.
  Further, by Lemma~\ref{lm:alg-codensity}, $\epsilon_k$ is the algebra
  induced by $\tau$ on the final $B$-coalgebra $B_k$. It is now easy to establish~\eqref{eq:mu-eps-iso}.
\end{proof}

%\begin{corollary}\label{cor:alg-codensity}
%	If the hypotheses of Theorem~\ref{thm:companion-codensity} are met, and $\C$ has an initial object $0$, then
%	$\cf{B} 0$ is the carrier of a final $B$-coalgebra, and
%	$\mu_0 \colon \cf{B} \cf{B} 0 \rightarrow \cf{B}0$ is the
%	algebra induced on it by $\tau$.
%\end{corollary}
%\begin{proof}
%	By Theorem~\ref{prop:monad-on-final}, $\eta$ and $\mu$
%	are the unique natural transformations such that $\tau$ is a distributive law
%	of $(\cf{B},\eta,\mu)$ over $B$. By Corollary~\ref{cor:final-0}, $\cf{B}0$ is the
%	carrier of a final coalgebra. Finally by Proposition~\ref{prop:alg-on-final},
%	$\mu_0$ is indeed the algebra induced on $\cf{B}0$.
%\end{proof}
%
%
%
%\begin{lemma}\label{lm:cont-eps-mu}
%	Let $B \colon \C \rightarrow \C$ be $\omega$-continuous,
%	and suppose $\C$ has an initial object $0$.
%	The algebra $\epsilon_\omega \colon \cf{B} B_\omega \rightarrow B_\omega$
%	is isomorphic to the algebra $\mu_0 \colon \cod{F}\cod{F} 0 \rightarrow \cod{F} 0$,
%	where $\mu$ is the multiplication of the codensity monad.
%\end{lemma}

\section{Codensity and the companion of a monotone function}\label{sec:cod-lattice}

Throughout this section, let $b \colon L \rightarrow L$ be a monotone
function on a complete lattice.  By
Theorem~\ref{thm:companion-codensity}, the companion of a monotone
function $b$ (viewed as a functor on a poset category) is given by the
right Kan extension of the final sequence
$\fin{b} \colon \Ord^{\op} \rightarrow L$ along itself. Using
Lemma~\ref{lm:pw-codensity}, we obtain the characterisation of the
companion given in the Introduction~(\ref{eq:pw16}).

\begin{theorem}
  The companion $t$ of a monotone function $b$ on a complete lattice is given by
  \[t: x\mapsto \bigwedge_{x \leq b_i} b_i\]
\end{theorem}
\begin{proof}
  By Lemma~\ref{lm:pw-codensity}, the codensity monad $\cod{\bar{b}}$
  can be computed by
  \[\cod{\bar{b}}(x) = \Ran{\fin{b}}{\fin{b}}(x) =  \bigwedge_{x \leq b_i} b_i\,,\]
  a limit that exists since $L$ is a complete lattice.  We apply
  Theorem~\ref{thm:companion-codensity} to show that $\cod{\bar{b}}$
  is the companion of $b$.  The preservation condition of the theorem
  amounts to the equality
  $b \circ \Ran{\fin{b}}{\fin{b}} = \Ran{\fin{b}}{b \circ \fin{b}}$
  which, by Lemma~\ref{lm:pw-codensity}, in turn amounts to
  \[
  b(\bigwedge_{x \leq b_i} b_i) = \bigwedge_{x \leq b_i} b(b_i)
  \]
  for all $x \in L$. The sequence ${(b_i)}_{i\in\Ord}$ is decreasing and
  stagnates at some ordinal $\epsilon$; therefore, the two
  intersections collapse into their last terms, say $b_k$ and
  $b(b_k)$ (with $k$ the greatest ordinal such that
  $x\not\leq b_{k+1}$, or $\epsilon$ if such an ordinal does
  not exist). The equality follows.
\end{proof}

In fact, the category $\K(b)$ defined in Section~\ref{sec:codensity}
instantiates to the following: an object is a monotone function
$f \colon L \rightarrow L$ such that $f(b_i) \leq b_i$ for all
$i \in \Ord$, and an arrow from $f$ to $g$ exists iff $f \leq g$.  The
companion $t$ is final in this category. This
yields the following characterisation of functions below the
companion.

\begin{proposition}\label{prop:cond-comp-mon}
  Let $t$ be the companion of a monotone function $b$ on a complete
  lattice.  For all monotone functions $f$ we have $f \leq t$ iff
  $\forall i \in \Ord: f(b_i) \leq b_i$.
\end{proposition}

A key intuition about up-to techniques is that they should at least
preserve the greatest fixpoint (i.e., up-to context is valid only when
bisimilarity is a congruence). It is however well-known that this is
not a sufficient condition~\cite{San98MFCS,SW01}. The above
proposition gives a stronger and better intuition: a technique should
preserve all approximations of the greatest fixpoint (the elements of
the final sequence) to be below the companion, and thus sound.

This intuition on complete lattices leads us to the abstract notion of
causality we introduce in the following section.

\section{Causality via right Kan extensions}\label{sec:causal}

We focus on the right Kan extension of the final sequence of $A$
along the final sequence of $B$, where both $A$ and $B$ are $\omega$-continuous $\Set$ endofunctors.
For such functors,
$A_\omega$ and $B_\omega$ are the carriers of the respective final coalgebras, and
Lemma~\ref{lm:codensity-set} provides us with a description of the
codensity monad in terms of natural transformations of the form
${(\fin{A}-)}^X \Rightarrow \fin{B}$. We show that such natural
transformations correspond to a new abstract notion which we call
\emph{causal operations}. Based on this correspondence and
Theorem~\ref{thm:companion-codensity}, we will get a concrete
understanding of the companion of $(A,B)$ in Section~\ref{sec:comp-set}.

\begin{definition}\label{def:causal-alg}
  Let $A,B,F\colon \Set \rightarrow \Set$ be functors.  An \emph{($\omega$)-causal operation} (from $A$ to $B$)
  is a map $\alpha \colon FA_\omega \rightarrow B_\omega$ such that for every set $X$, functions
  $f,g \colon X \rightarrow A_\omega$ and $i < \omega$, we have
  that $A_{\omega, i} \circ f = A_{\omega,i} \circ g$ implies
  $B_{\omega,i} \circ \alpha \circ Ff = B_{\omega,i} \circ \alpha \circ Fg$, i.e.,
  the commutativity of the diagram on the left-hand side below implies commutativity
  of the right-hand side.
  \[
  \xymatrix@R=0.5cm{
    & A_\omega \ar[dr]^{A_{\omega,i}} & \\ % chktex 3
    X \ar[ur]^{f} \ar[dr]_g % chktex 3
    & & A_i \\
    & A_\omega \ar[ur]_{A_{\omega,i}} % chktex 3
  } \quad \quad \quad
  \xymatrix@R=0.5cm{
    & FA_\omega \ar[r]^{\alpha} & B_\omega \ar[dr]^{B_{\omega,i}} & \\ % chktex 3
    FX \ar[ur]^{Ff} \ar[dr]_{Fg} % chktex 3
    & & & B_i \\
    & FA_\omega \ar[r]_{\alpha} % chktex 3
    & B_\omega \ar[ur]_{B_{\omega,i}} % chktex 3
  }
  \]
  If $A=B$ then we refer to such $\alpha$ as a \emph{causal algebra}.
  Further, an \emph{($\omega$)-causal function} on $|V|$ arguments is a causal
  operation where $F = {(-)}^V$.  Equivalently,
  $\alpha \colon {(A_\omega)}^V \rightarrow B_\omega$ is a causal function iff for
  every $h,k \in {(A_\omega)}^V$ and every $i < \omega$:
  \[
  A_{\omega,i} \circ h = A_{\omega,i} \circ k \quad \text{implies} \quad
  B_{\omega,i} (\alpha (h)) = B_{\omega,i} (\alpha(k)) \,.
  \]
  Causal operations form a category $\causal{A,B}$: an object is a pair
  $(F,\alpha \colon FA_\omega \rightarrow B_\omega)$ where $\alpha$ is
  causal, and a morphism from $(F,\alpha)$ to $(G,\beta)$ is a natural
  transformation $\kappa \colon F \Rightarrow G$ such that
  $\beta \circ \kappa_{A_\omega} = \alpha$.
\end{definition}

\begin{example}\label{ex:causal-streams}
  Recall from Example~\ref{ex:finseq-streams} that, for the functor
  $BX = \real\times X$, $B_i$ is the set of lists of length $i$, and in
  particular $B_\omega$ is the set of streams over $\real$.  We focus
  first on causal \emph{functions}. To this end, for
  $\sigma,\tau \in B_\omega$, we write $\sigma \equiv_i \tau$ if
  $\sigma$ and $\tau$ are equal up to $i$, i.e., $\sigma(k) = \tau(k)$
  for all $k < i$. It is easy to verify that a function of the form
  $\alpha \colon {(B_\omega)}^n \rightarrow B_\omega$ is causal iff for
  all $\sigma_1, \ldots, \sigma_n, \tau_1, \ldots, \tau_n$ and all
  $i < \omega$: if $\sigma_j \equiv_i \tau_j$ for all $j \leq n$ then
  $\alpha(\sigma_1, \ldots, \sigma_n) \equiv_i
  \alpha(\tau_1,\ldots,\tau_n)$.

  For instance, taking $n=2$,
  $\mathsf{alt}(\sigma,\tau) = (\sigma(0),\tau(1),\sigma(2),\tau(3),
  \ldots)$
  is causal, whereas
  $\even(\sigma) = (\sigma(0), \sigma(2), \ldots)$ (with
  $n=1$) is not causal. Standard operations from the
  stream calculus~\cite{Rutten05}, such as pointwise stream addition and shuffle product,
  are causal.

  The above notion of causal function (with a finite set of arguments
  $V$) agrees with the standard notion of causal stream function
  (e.g.,~\cite{HansenKR16}). Our notion of causal \emph{algebras}
  generalises it from single functions to algebras for arbitrary
  functors. This includes polynomial functors modelling a
  signature.

  For $A = \real$,
  the algebra
  $\alpha \colon \powf(B_\omega) \rightarrow B_\omega$ for the finite powerset functor $\powf$, defined by
  $\alpha(S)(n) = \min \{\sigma(n) \mid \sigma \in S\}$, is a
  causal algebra which is not a causal function.
  The algebra $\beta \colon \powf(B_\omega) \rightarrow B_\omega$ given by
  $\beta(S)(n) = \sum_{\sigma \in S} \sigma(n)$ is \emph{not} causal
  according to Definition~\ref{def:causal-alg}. Intuitively,
  $\beta(\{\sigma,\tau\})(i)$ depends on equality of $\sigma$ and $\tau$, since addition of real numbers is not idempotent.
\end{example}

\begin{example}\label{ex:causal-even}
  Let $BX = \real \times X$ as above, and
  $AX = \real \times \real \times X$.  Then $A_i$ is the set of lists
  over $\real \times \real$ of length $i$, isomorphic to the set of
  lists over $\real$ of length $2i$. In particular, $A_\omega$ is
  (isomorphic to) the set $\real^\omega$ of all lists over $\real$.
  The projections $A_{\omega,i} \colon A_\omega \rightarrow A_i$ map a
  stream to the first $2i$ elements.

  An $n$-ary causal function from $A$ to $B$ is a map
  $\alpha \colon {(\real^\omega)}^n \rightarrow \real^\omega$ such that
  for all $\sigma_1, \ldots, \sigma_n, \tau_1, \ldots, \tau_n$ and all
  $i < \omega$: if $\sigma_j \equiv_{2i} \tau_j$ for all $j \leq n$
  then
  $\alpha(\sigma_1, \ldots, \sigma_n) \equiv_i
  \alpha(\tau_1,\ldots,\tau_n)$. For instance, the function $\even$
  described in the previous example is a causal function from $A$ to
  $B$ (but not from $B$ to $B$).

  The function $\double \colon \real^\omega \rightarrow \real^\omega$
  given by
  $\double(\sigma) = (\sigma(0),\sigma(0), \sigma(1), \sigma (1),
  \ldots)$ is causal from $B$ to $A$. It is easy to check that causal
  functions compose, so that $\even \circ \double$ is causal from $B$
  to $B$, and $\double \circ \even$ is causal from $A$ to $A$.  Such a
  behaviour is reminiscent of the sized-types approach to
  coinductive data-types in type theory~\cite{AbelP13}.
\end{example}

\begin{example}
  Let $BX = X + 1$. Recall the characterisation
  $B_i = \{0, 1, \ldots, i\}$ from Example~\ref{ex:finseq-plusone}; in
  particular $B_\omega = \mathbb{N} \cup \{\omega\}$, ordered as
  usual, with $\omega$ the top element.  According to
  Definition~\ref{def:causal-alg}, a function
  $f \colon B_\omega \rightarrow B_\omega$ is causal if for all
  $x,y \in B_\omega$, $i \in \mathbb{N}$:
  \[
  (x=y \vee i \leq \min(x, y)) \rightarrow (f(x) = f(y) \vee i
  \leq \min(f(x), f(y))) \,.\] It can be shown that this is equivalent
  to:
  \[
  \forall x,y \in B_\omega .\, (x > f(x) \wedge y > f(x))
  \rightarrow f(x) = f(y) \,.
  \]

  Now let $AX = \real \times X$, and $B$ as above. Using the notation
  $\equiv_i$ from Example~\ref{ex:causal-streams}, a function
  $f \colon \real^\omega \rightarrow \mathbb{N} \cup \{\omega\}$ is
  causal (from $A$ to $B$) if for all $\sigma, \tau \in \real^\omega$
  and all $i \in \mathbb{N}$:
  \[
  \sigma \equiv_i \tau \rightarrow (f(\sigma) = f(\tau) \vee i \leq \min(f(\sigma),f(\tau))) \,.
  \]
  For instance, the function
  $f(\sigma) = \min\{i \mid \sigma(i) = 0 \}$ which computes the
  position of the first zero in $\sigma$, or $\omega$ if zero does
  not appear, is causal.
\end{example}

\begin{example}
  For the functor $BX = 2 \times X^A$, $B_\omega = \pow(A^*)$ is the
  set of languages over $A$ (Example~\ref{ex:finseq-automata}).  Given
  languages $L$ and $K$, we write $L \equiv_i K$ if $L$ and $K$
  contain the same words of length below $i$.  A function
  $\alpha \colon {(\pow(A^*))}^n \rightarrow \pow(A^*)$ is causal iff
  for all languages $L_1, \ldots, L_n, K_1, \ldots, K_n$: if
  $L_j \equiv_i K_j$ for all $j \leq n$ then
  $\alpha(L_1, \ldots, L_n) \equiv_i \alpha(K_1,\ldots,K_n)$. For
  instance, union, concatenation, Kleene star, and shuffle of
  languages are all causal.  An example of a causal algebra which is
  not a causal function is
  $\alpha \colon \pow(\pow(A^*)) \rightarrow \pow(A^*)$ defined by
  union.
\end{example}

The following result connects causal operations to natural
transformations of the form $F \fin{A} \Rightarrow \fin{B}$ (which,
from Section~\ref{sec:codensity}, form a category $\K(\fin{A},\fin{B})$).

\begin{theorem}\label{thm:causal-alg}
  Let $A,B,F \colon \Set \rightarrow \Set$ be functors, and suppose $A$ and $B$ are
  $\omega$-continuous.  The category $\causal{A,B}$ of causal
  operations is isomorphic to the category $\K(\fin{A},\fin{B})$. Concretely,
  there is a one-to-one correspondence
  between natural transformations $\alpha \colon F \fin{A} \Rightarrow \fin{B}$
  and causal algebras $\alpha_\omega \colon FA_\omega \rightarrow B_\omega$.
  \[
  \frac{\alpha \colon F \fin{A} \Rightarrow \fin{B}}
  {\alpha_\omega \colon FA_\omega \rightarrow B_\omega \qquad \text{causal}}
  \]
  From top to bottom, this is given by evaluation at $\omega$.
%  Moreover, we have $\beta \circ \kappa {\fin{B}} = \alpha$ iff
%  $\beta_\omega \circ \kappa_{B_\omega} = \alpha_\omega$ for any
%  $\alpha \colon F \fin{B} \Rightarrow \fin{B}$,
%  $\beta \colon G \fin{B} \Rightarrow \fin{B}$ and
%  $\kappa \colon F \Rightarrow G$.
  Moreover, for any $\alpha \colon F \fin{A} \Rightarrow \fin{B}$,
$\beta \colon G \fin{A} \Rightarrow \fin{B}$ and
$\kappa \colon F \Rightarrow G$, we have
$\beta \circ \kappa \fin{A} = \alpha$ (as on the left below) iff
  $\beta_\omega \circ \kappa {A_\omega} = \alpha_\omega$ (as on the right).
\[
\xymatrix{
	F\fin{A} \ar[rr]^{\kappa \fin{A}} \ar[dr]_{\alpha} % chktex 3
		& & G\fin{A} \ar[dl]^{\beta} \\ % chktex 3
		& \fin{B}
}
\quad
\quad
\xymatrix{
	FA_\omega \ar[rr]^{\kappa_{A_\omega}} \ar[dr]_{\alpha_\omega} % chktex 3
		& & GA_\omega \ar[dl]^{\beta_\omega} \\ % chktex 3
		& B_\omega
}
\]
%for any $\alpha \colon F \fin{A} \Rightarrow \fin{B}$,
%$\beta \colon G \fin{A} \Rightarrow \fin{B}$ and
%$\kappa \colon F \Rightarrow G$.
\end{theorem}

\begin{proof}
	Let $\alpha \colon F\fin{A} \Rightarrow \fin{B}$. We need to prove that $\alpha_\omega$ is causal; to this end,
	let $f,g \colon X \rightarrow FA_\omega$ be functions such
	that $A_{\omega,i} \circ f = A_{\omega,i} \circ g$ for some $i$.
	Then the following diagram commutes:
	\[
	\xymatrix{
		& FA_\omega \ar[r]^{\alpha_\omega} \ar[dr]_{FA_{\omega,i}} & B_\omega \ar[dr]^{B_{\omega,i}} & \\ % chktex 3
	FX \ar[ur]^{Ff} \ar[dr]_{Fg} % chktex 3
		& & FA_i \ar[r]^{\alpha_i} & B_i \\ % chktex 3
		& FA_\omega \ar[ur]^{FA_{\omega,i}} \ar[r]_{\alpha_\omega} % chktex 3
		& B_\omega \ar[ur]_{B_{\omega,i}} % chktex 3
	}
	\]
	by assumption and naturality of $\alpha$. Hence $\alpha_\omega$ is causal.

	Next, we show how to define a natural transformation $\alpha \colon F\fin{A} \Rightarrow \fin{B}$ from a given causal $\alpha_\omega \colon FA_\omega \rightarrow B_\omega$.
	Since $A$ is $\omega$-continuous (similarly for $B$) and any $\Set$ endofunctor preserves epimorphisms, one can prove
	by induction that for any $i<\omega$, the map $A_{\omega,i}$ is an epi. We will use that epis
	in $\Set$ split, i.e., every $A_{\omega,i}$ has a right inverse $A_{\omega,i}^{-1}$
	with $A_{\omega,i} \circ A_{\omega,i}^{-1} = \id$.

	Given $\alpha_\omega \colon FA_\omega \rightarrow B_\omega$, define $\alpha \colon F\bar{A} \Rightarrow \bar{B}$
	on a component $i < \omega$ by
	\[
	\xymatrix{
		FA_i \ar[r]^{F(A_{\omega,i}^{-1})} % chktex 3
			& FA_\omega \ar[r]^{\alpha_\omega} % chktex 3
			& B_\omega \ar[r]^{B_{\omega,i}} % chktex 3
			& B_i
	}
	\]
	where $A_{\omega,i}^{-1}$ is a right inverse of $A_{\omega,i}$. It follows from causality of $\alpha_\omega$
	that the choice of right inverse is irrelevant. On
	a component $i \geq \omega$, $\alpha$ is defined by
	\[
	\xymatrix{
		FA_i \ar[r]^{FA_{i,\omega}} % chktex 3
			& FA_\omega \ar[r]^{\alpha_\omega} % chktex 3
			& B_\omega \ar[r]^{B_{i,\omega}^{-1}} % chktex 3
			& B_i
	}
	\]
	where $B_{i,\omega}^{-1}$ is the inverse of $B_{i,\omega}$ (which is an isomorphism, since $B$ is $\omega$-continuous).

	We need to show that $\alpha$ is a natural transformation, and that the correspondence is bijective.
	For the bijective correspondence, first note that mapping $\alpha_\omega$ to $\alpha$ and back
	trivially yields $\alpha_\omega$ again. Conversely, given $\alpha$, we need to prove that the
	following diagrams commute for $i<\omega$ (on the left)
	and $i \geq \omega$ (on the right):
	\[
	\xymatrix{
	FA_i \ar[r]^{\alpha_i} \ar[d]_{F(A_{\omega,i}^{-1})} % chktex 3
		& B_i \\ % chktex 3
	FA_\omega \ar[r]_{\alpha_\omega} % chktex 3
		& B_\omega \ar[u]_{B_{\omega,i}} % chktex 3
	}
	\qquad
	\xymatrix{
	FA_i \ar[r]^{\alpha_i} \ar[d]_{FA_{i,\omega}} % chktex 3
		& B_i \\
	FA_\omega \ar[r]_{\alpha_\omega} % chktex 3
		& B_\omega \ar[u]_{B_{i,\omega}^{-1}} % chktex 3
	}
	\]
	The case $i<\omega$ follows by naturality of $\alpha$ and since $A_{\omega,i}^{-1}$ is a right inverse of $A_{\omega,i}$,
	the case $i \geq \omega$ by naturality of $\alpha$ and since $B_{i,\omega}^{-1}$ is a (left) inverse of $B_{i,\omega}$.

	It remains to be shown that $\alpha$, defined from a given $\alpha_\omega$ as above, is natural, using that
	$\alpha_\omega$ is causal. To this end, let $i \leq j$; to prove is that the following diagram commutes:
	\[
	\xymatrix{
		FA_i \ar[r]^{\alpha_i} % chktex 3
			& B_i \\ % chktex 3
		FA_j \ar[r]_{\alpha_j} \ar[u]^{FA_{j,i}} % chktex 3
			& B_j \ar[u]_{B_{i,j}} % chktex 3
	}
	\]
	where $\alpha_i$, $\alpha_j$ are defined from $\alpha_\omega$ as above.
	We proceed with a case distinction.

	If $i,j < \omega$, then the following diagram commutes:
	\[
	\xymatrix{
		A_j \ar[r]^{A_{j,i}} \ar[d]_{A_{\omega,j}^{-1}} \ar@{=}[dr] % chktex 3
			& A_i\ar[r]^{A_{\omega,i}^{-1}} \ar@{=}[dr] % chktex 3
			& A_\omega \ar[d]^{A_{\omega,i}} \\ % chktex 3
		A_\omega \ar[r]^{A_{\omega,j}} \ar@/_15pt/[rr]_{A_{\omega,i}} % chktex 3 chktex 25
			& A_j \ar[r]^{A_{j,i}} % chktex 3
			& A_i
	}
	\]
	since $A_{\omega,i}^{-1}$ and $A_{\omega,j}^{-1}$ are right inverses (for the two triangles)
	and the final sequence $\bar{A}$ is a functor (for the crescent). By causality of $\alpha_\omega$
	(and functoriality of $\bar{B}$) we obtain commutativity of:
	\[
	\xymatrix{
		FA_i \ar[r]^{F(A_{\omega,i}^{-1})} % chktex 3
		  & FA_\omega \ar[r]^{\alpha_\omega} % chktex 3
		  & B_\omega \ar[r]^{B_{\omega,i}} % chktex 3
		  & B_i \\
		FA_j \ar[u]^{FA_{j,i}} \ar[r]_{F(A_{\omega,j}^{-1})} % chktex 3
			& FA_\omega \ar[r]_{\alpha_\omega} % chktex 3
			& B_\omega \ar[ur]^(0.4){B_{\omega,i}} \ar[r]_{B_{\omega,j}} % chktex 3
			& B_j \ar[u]_{B_{j,i}} % chktex 3
	}
	\]
	which is what we needed to prove, by definition of $\alpha_i$ and $\alpha_j$.

	If $i < \omega \leq j$, then the following diagram commutes:
	\[
	\xymatrix{
		A_i \ar[r]^{A_{\omega,i}^{-1}} % chktex 3
			& A_\omega \ar[r]^{A_{\omega,i}} % chktex 3
			& A_i \\
		A_j \ar[rr]_{A_{j,\omega}} \ar[u]^{A_{j,i}} % chktex 3
			& & A_\omega \ar[u]_{A_{\omega,i}} % chktex 3
	}
	\]
	since $A_{\omega,i} \circ A_{\omega,i}^{-1} = \id$, and the final sequence $\bar{A}$ is a functor. Hence,
	by causality of $\alpha$ we obtain
	the commutativity of the large inner part in:
	\[
	\xymatrix{
		FA_i \ar[r]^{F(A_{\omega,i}^{-1})} % chktex 3
		  & FA_\omega \ar[r]^{\alpha_\omega} % chktex 3
		  & B_\omega \ar[r]^{B_{\omega,i}} % chktex 3
		  & B_i \\
		FA_j \ar[u]^{FA_{j,i}} \ar[r]_{FA_{j,\omega}} % chktex 3
			& FA_\omega \ar[r]_{\alpha_\omega} % chktex 3
			& B_\omega \ar[ur]^(0.4){B_{\omega,i}} \ar[r]_{B_{j,\omega}^{-1}} % chktex 3
			& B_j \ar[u]_{B_{j,i}} % chktex 3
	}
	\]
	The triangle commutes since $\bar{B}$ is functorial and $B_{j,\omega}^{-1}$ is an inverse
	of $B_{j,\omega}$.

	Finally, if $\omega \leq i \leq j$, then we immediately obtain commutativity of:
	\[
	\xymatrix{
		FA_i \ar[r]^{FA_{i,\omega}} % chktex 3
		  & FA_\omega \ar[r]^{\alpha_\omega} % chktex 3
		  & B_\omega \ar[r]^{B_{i,\omega}^{-1}} % chktex 3
		  & B_i \\
		FA_j \ar[u]^{FA_{j,i}} \ar[r]_{FA_{j,\omega}} \ar[ur]_(0.6){FA_{j,\omega}} % chktex 3
			& FA_\omega \ar[r]_{\alpha_\omega} % chktex 3
			& B_\omega \ar[ur]^(0.4){B_{i,\omega}^{-1}} \ar[r]_{B_{j,\omega}^{-1}} % chktex 3
			& B_j \ar[u]_{B_{j,i}} % chktex 3
	}
	\]
	The triangles commute by functoriality of $\bar{A},\bar{B}$ and the fact that $B_{i,\omega}^{-1}$
	and $B_{j,\omega}^{-1}$ are inverses of $B_{i,\omega}$ and $B_{j,\omega}$ respectively.

	This concludes the one-to-one correspondence between natural transformations $\alpha$ and
	causal algebras $\alpha_\omega$. We turn to the second correspondence in the statement:
	the equivalence \[\beta \circ \kappa {\fin{A}} = \alpha \qquad \text{ iff } \qquad
  \beta_\omega \circ \kappa_{A_\omega} = \alpha_\omega\]
  for any
  $\alpha \colon F \fin{A} \Rightarrow \fin{B}$,
  $\beta \colon G \fin{A} \Rightarrow \fin{B}$ and
  $\kappa \colon F \Rightarrow G$.
%	  $$
%\xymatrix{
%	F\fin{B} \ar[rr]^{\kappa \fin{B}} \ar[dr]_{\alpha}
%		& & G\fin{B} \ar[dl]^{\beta} \\
%		& \fin{B}
%}
%\quad \text{  }
%\quad
%\xymatrix{
%	FB_\omega \ar[rr]^{\kappa_{B_\omega}} \ar[dr]_{\alpha_\omega}
%		& & GB_\omega \ar[dl]^{\beta_\omega} \\
%		& B_\omega
%}
%$$
%for any $\alpha \colon F \fin{B} \Rightarrow \fin{B}$,
%$\beta \colon G \fin{B} \Rightarrow \fin{B}$ and
%$\kappa \colon F \Rightarrow G$.
From left to right this is trivial. Conversely, suppose $\beta_\omega \circ \kappa_{A_\omega} = \alpha_\omega$.
By the above, both $\alpha$ and $\beta$ extend to natural transformations.
First, suppose $i < \omega$.
%	$\beta_\omega$ extends to a natural transformation $\beta \colon G\fin{A} \Rightarrow \fin{B}$.
	%
	We need to prove that the outside of the following diagram commutes:
	\[
	\xymatrix{
		FA_i \ar[rrr]^{\kappa_{A_i}} \ar[dr]_(0.6){FA_{\omega,i}^{-1}} \ar[ddd]_{\alpha_i} % chktex 3
			& & & GA_i\ar[ddd]^{\beta_i}\\ % chktex 3
			& FA_\omega \ar[d]_{\alpha_\omega} \ar[r]^{\kappa_{A_\omega}} % chktex 3
			& GA_\omega \ar[d]^{\beta_\omega} \ar[ur]_(0.4){GA_{\omega,i}} % chktex 3
			& \\
			& B_\omega \ar[dl]_(0.4){B_{\omega,i}} \ar@{=}[r] % chktex 3
			& B_\omega \ar[dr]^(0.4){B_{\omega,i}} \\ % chktex 3
			B_i \ar@{=}[rrr]
			& & & B_i
	}
	\]
	The middle square commutes by assumption. The rest, clockwise starting at the north,
	by naturality of $\kappa$ and $A_{\omega,i}^{-1}$ being a right inverse of $A_{\omega,i}$,
	naturality of $\beta$, trivially, and by naturality of $\alpha$
	and $A_{\omega,i}^{-1}$ being a right inverse of $A_{\omega,i}$.

	For $i \geq \omega$, we have
	$\alpha_i = B_{i,\omega}^{-1} \circ \alpha_\omega \circ FA_{i,\omega}$ and
	$\beta_i = B_{i,\omega}^{-1} \circ \beta_\omega \circ FA_{i,\omega}$, hence
	it suffices to prove commutativity of the diagram below.
	\[
	\xymatrix{
		FA_i \ar[d]_{FA_{i,\omega}} \ar[rr]^{\kappa_{A_i}} % chktex 3
			& & GA_i \ar[d]^{GA_{i,\omega}}\\ % chktex 3
		FA_\omega \ar[rr]^{\kappa_{A_\omega}} \ar[d]_{\alpha_\omega} % chktex 3
			& & GA_\omega \ar[d]^{\beta_\omega} \\ % chktex 3
		B_\omega \ar[dr]_{B_{i,\omega}^{-1}} \ar@{=}[rr] % chktex 3
			& & B_\omega \ar[dl]^{B_{i,\omega}^{-1}} \\ % chktex 3
			& B_i &
	}
	\]
	That follows,
	again, from naturality of $\kappa$ and the assumption.
%From left to right this is trivial. Conversely, suppose $\beta_\omega \circ \kappa_{A_\omega} = \alpha_\omega$.
%Let $i < \omega$. By the above, we have
%	$\alpha_i = B_{\omega,i} \circ \alpha_\omega \circ FA_{\omega,i}^{-1}$
%	and $\beta_i = B_{\omega,i} \circ \beta_\omega \circ FA_{\omega,i}^{-1}$,
%	for any right inverse $A_{\omega,i}^{-1}$ of $A_{\omega,i}$.
%
%	Hence, it suffices
%	to prove that the diagram on the left below commutes:
%	$$
%	\xymatrix{
%		FA_i \ar[d]_{F(A_{\omega,i}^{-1})} \ar[rr]^{\kappa_{A_i}}
%			& & GA_i \ar[d]^{G(A_{\omega,i}^{-1})}\\
%		FA_\omega \ar[rr]^{\kappa_{A_\omega}} \ar[d]_{\alpha_\omega}
%			& & GA_\omega \ar[d]^{\beta_\omega} \\
%		B_\omega \ar[dr]_{B_{\omega,i}} \ar@{=}[rr]
%			& & B_\omega \ar[dl]^{B_{\omega,i}} \\
%			& B_i &
%	}
%	\qquad
%	\xymatrix{
%		FA_i \ar[d]_{FA_{i,\omega}} \ar[rr]^{\kappa_{A_i}}
%			& & GA_i \ar[d]^{GA_{i,\omega}}\\
%		FA_\omega \ar[rr]^{\kappa_{A_\omega}} \ar[d]_{\alpha_\omega}
%			& & GA_\omega \ar[d]^{\beta_\omega} \\
%		B_\omega \ar[dr]_{B_{i,\omega}^{-1}} \ar@{=}[rr]
%			& & B_\omega \ar[dl]^{B_{i,\omega}^{-1}} \\
%			& B_i &
%	}
%	$$
%	The upper square of the left diagram commutes by naturality of $\kappa$, and the lower by assumption.
%	For $i \geq \omega$, we have
%	$\alpha_i = B_{i,\omega}^{-1} \circ \alpha_\omega \circ FA_{i,\omega}$ and
%	$\beta_i = B_{i,\omega}^{-1} \circ \beta_\omega \circ FA_{i,\omega}$, hence
%	it suffices to prove commutativity of the diagram on the right above. That follows,
%	again, from naturality of $\kappa$ and the assumption.
\end{proof}

By the above theorem, the universal property of the right Kan extension
amounts to the following property of causal algebras.

\begin{corollary}\label{cor:final-causal-alg}
  Suppose $A,B \colon \Set \rightarrow \Set$ are $\omega$-continuous.
  Let $\epsilon$ be the counit of $\kfab$.  Then $\epsilon_\omega$ is
  final in $\causal{A,B}$, i.e., $\epsilon_\omega$ is causal and for every causal algebra
  $\alpha_\omega \colon FA_\omega \rightarrow B_\omega$, there is a unique
  natural transformation $\hat{\alpha} \colon F \Rightarrow \kfab$
  such that $\epsilon_\omega \circ \hat{\alpha}_{A_\omega} = \alpha_\omega$.
  \[
  \xymatrix{
    FA_\omega \ar[rr]^{\hat{\alpha}_{A_\omega}} \ar[dr]_{\alpha_\omega} & % chktex 3
    & \kfab A_\omega \ar[dl]^{\epsilon_\omega} \\ % chktex 3
    & B_\omega & }
  \]
\end{corollary}

%By Lemma~\ref{lm:cont-eps-mu}, we can phrase the above in terms of the multiplication of the codensity monad.
%\begin{corollary}\label{cor:final-causal-mult}
%	Let $B \colon \Set \rightarrow \Set$ be $\omega$-continuous.
%	Let $\mu$ be the multiplication of the codensity monad $\cf{B}$. Then
%	$\mu_\emptyset$ is a final object in the category of causal algebras (modulo the isomorphism
%	 $B_\omega \cong \cf{B}(\emptyset)$).
%\end{corollary}
By Lemma~\ref{lm:codensity-set} and Theorem~\ref{thm:causal-alg}, we
obtain the following concrete description of the relevant right Kan extension
as a \emph{functor of causal functions}.

\begin{theorem}\label{thm:causalfunctions}
  Let $A,B \colon \Set \rightarrow \Set$ be $\omega$-continuous
  functors. The right Kan extension $\kfab$ is given by
  \begin{align*}
    \kfab(X) &= \{\alpha \colon A_\omega^X \rightarrow B_\omega \mid \alpha \text{ is a causal function}\}\,,\\
    \kfab(h \colon X \rightarrow Y)(\alpha) &= \lambda f. \,\alpha(f \circ h)\,,
  \end{align*}
  and, for the counit
  $\epsilon \colon \kfab\fin{A} \Rightarrow \fin{B}$, we have
  $\epsilon_\omega(\alpha \colon A_\omega^{A_\omega} \rightarrow
  B_\omega) = \alpha(\id_{A_\omega})$.

  Finally, given $\alpha \colon F\fin{A} \Rightarrow \fin{B}$,
  the unique natural transformation $\hat{\alpha} \colon F \Rightarrow \kfab$
%  such that $\epsilon_\omega \circ \hat{\alpha}_{A_\omega} = \alpha$ is given
such that $\epsilon \circ \hat{\alpha} \fin{A} = \alpha$ is given
  by $\hat{\alpha}_X(S) \colon A_\omega^X \rightarrow B_\omega$, $f \mapsto \alpha_\omega \circ Ff(S)$.
  Equivalently, $\hat{\alpha}$ is the unique natural transformation
  such that $\epsilon_\omega \circ \hat{\alpha}_{A_\omega} = \alpha_\omega$.
\end{theorem}

Taking $A = B = \real \times \Id$, the above right Kan extension (codensity monad)
maps a set $X$ to set of all causal stream functions with $|X|$ arguments. Similarly, for the
functor $X \mapsto 2 \times X^A$ we obtain a functor of causal
functions on languages.

\begin{example}
The right Kan extension $\kfab$ has a universal property in the category of causal algebras,
via Corollary~\ref{cor:final-causal-alg}. We give an example for the case of streams, taking $A = B = \real \times \Id$,
so the relevant right Kan extension $\kfab$ is the codensity monad $\cf{B}$.
Consider a causal algebra on streams, of the form $\alpha_\omega  \colon \real^\omega \times \real^\omega \rightarrow \real^\omega$ for the functor $SX = X^2$.
It follows from Theorem~\ref{thm:causalfunctions}
that the unique natural transformation $\hat{\alpha} \colon S \Rightarrow \cf{B}$
such that $\epsilon_\omega \circ \hat{\alpha}_{\real^\omega} = \alpha_\omega$
is
given on an $X$-component
$\hat{\alpha}_X \colon X \times X \rightarrow \cf{B}(X)$ by
\[\hat{\alpha}_X(x,y)(f \colon X \rightarrow \real^\omega) = \alpha_\omega(f(x),f(y)) \, .\]
%To see
%this, notice that for all $\sigma, \tau \in \real^\omega$, we have: $\epsilon_\omega ( \hat{\alpha}_{\real^\omega}(\sigma,\tau)) = \hat{\alpha}_{\real^\omega}(\sigma,\tau)(\id_{\real^\omega})
% = \alpha(\id_{\real_\omega}(\sigma), \id_{\real^\omega}(\tau)) = \alpha(\sigma,\tau)$.

%By Corollary~\ref{cor:final-causal-alg} there is a unique natural transformation $\hat{\alpha} \colon S \Rightarrow \cf{B}$
%such that $\epsilon_\omega \circ \hat{\alpha}_{\real^\omega} = \alpha$. It is defined on an $X$-component
%$\hat{\alpha}_X \colon X \times X \rightarrow \cf{B}(X)$ by
%$$\hat{\alpha}_X(x,y)(f \colon X \rightarrow \real^\omega) = \alpha(f(x),f(y)) \, .$$
%To see
%this, notice that for all $\sigma, \tau \in \real^\omega$, we have: $\epsilon_\omega ( \hat{\alpha}_{\real^\omega}(\sigma,\tau)) = \hat{\alpha}_{\real^\omega}(\sigma,\tau)(\id_{\real^\omega})
% = \alpha(\id_{\real_\omega}(\sigma), \id_{\real^\omega}(\tau)) = \alpha(\sigma,\tau)$.
\end{example}

\section{Companion of polynomial \texorpdfstring{$\Set$}{Set} functors}\label{sec:comp-set}

The previous sections give us a concrete understanding of the
codensity monad of the final sequence of a $\Set$ functor in terms of
causal functions, and Theorem~\ref{thm:companion-codensity} provides
us with a sufficient condition for this codensity monad to be the
companion. We now focus on several applications of these results.

A rather general class of functors that satisfy the hypotheses of
Theorem~\ref{thm:companion-codensity} is given by the \emph{polynomial
  functors}. Automata, stream systems, Mealy and Moore machines,
various kinds of trees, and many more are all examples of coalgebras
for polynomial functors (e.g.,~\cite{Jacobs:coalg}).  A functor
$B \colon \Set \rightarrow \Set$ is called polynomial (in a single
variable) if it is isomorphic to a functor of the form
\[
X \mapsto \coprod_{a \in A} X^{B_a}
\]
for some $A$-indexed collection ${(B_a)}_{a \in A}$ of sets. As
explained in~\cite[1.18]{gambino2013polynomial}, a $\Set$ functor $B$
is polynomial if and only if it preserves connected limits. This
implies existence and preservation by $B$ of the codensity monad of
$\fin{B}$, as required by Theorem~\ref{thm:companion-codensity}.
\begin{lemma}\label{lm:connected}
  If $B \colon \Set \rightarrow \Set$ is polynomial, then $\cf B$
  exists and $B$ preserves it.
\end{lemma}
For the proof, we first recall that a category is \emph{connected}~\cite{mac1998categories} if
it is inhabited and there is a zigzag of morphisms between any two objects $X$ and $Y$:
a finite collection of morphisms of the form
\[
\xymatrix{
	X = X_0 \ar[r]
		& X_1
		& \ar [l] X_2 \ar[r]
		& X_3
		& \ar[l] X_4 \ar[r]
		& \ldots
		& \ar[l]X_n = Y \,.
}
\]
A \emph{connected limit} is a limit over a connected category.
\begin{proof}
	Since $B$ is polynomial, $B_{k,\omega}$ is an isomorphism for each $k \geq \omega$, which
	implies that the category $\Delta_X / \fin{B}$ is essentially small for every set $X$.
	Hence, the limit
	\begin{equation}\label{eq:lim-in-proof}
			\lim\left((\Delta_X / \fin{B}) \rightarrow \Ord^\op \xrightarrow{\fin{B}} \Set \right)
	\end{equation}
	exists for each $X$, which, by Lemma~\ref{lm:pw-codensity}, defines the codensity monad $\cf{B}$.
	As mentioned above, since $B$ is polynomial, it preserves connected limits~\cite{gambino2013polynomial}.
	%~\cite[Proposition 1.16]{gambino2013polynomial}.
	We show that $\Delta_X / \fin{B}$ is connected:
	$\Delta_X / \fin{B}$ is inhabited since there is the arrow $!_X \colon X \rightarrow B_0 = 1$; and
	 for any $f \colon X \rightarrow B_i$, there is the arrow $B_{i,0}$ to $!_X \colon X \rightarrow B_0$, which is a morphism in
	$\Delta_X / \fin{B}$ by uniqueness. Hence, $B$ preserves the limits in~\eqref{eq:lim-in-proof}. This implies that $B$ preserves $\cf{B}$, which we describe in detail.

	Denote, for a given set $X$, the limiting cone of~\eqref{eq:lim-in-proof} by
    \[
        {\{s^X_f \colon \cf{B} X \rightarrow B_i\}}_{f \in B_i^X, i \in \Ord}.
    \]
	The counit of the codensity monad is defined by $\epsilon_i = s_{\id_{B_i}}^{B_i}$ (see, e.g.,~\cite{nlabrkan,mac1998categories}).
	Since $B$ preserves these limits, for each $X$, we have that
	\[
        {\{Bs^X_f \colon B\cf{B} X \rightarrow BB_i\}}_{f \in B_i^X, i \in \Ord}
    \]
    is the limit
	\[
	\lim\left((\Delta_X / \fin{B}) \rightarrow \Ord^\op \xrightarrow{\fin{B}} \Set \xrightarrow{B} \Set \right) \,.
	\]
	Hence, by Lemma~\ref{lm:pw-codensity}, $B\cf{B}$ is a right Kan extension of $B\fin{B}$ along $\fin{B}$, with counit defined on $i \in \Ord$ by
	$Bs^{B_i}_{\id_{B_i}} = B\epsilon_i$ as desired.
\end{proof}

As a consequence, if $B$ is polynomial, the functor of causal
functions in Theorem~\ref{thm:causalfunctions} is the companion of
$B$.

\subsection{Causal algebras and coinduction up-to}%
\label{sec:causal-coinduction}

As explained in Section~\ref{sec:dl}, a distributive law of $F$ over $B$
allows one to strengthen the coinduction principle, formalised in terms of
$BF$-coalgebras, leading to an expressive coinductive definition (and proof)
technique. This approach is formally supported by Theorem~\ref{thm:valid-bartels}.  Based
on the characterisation of the companion in terms of causal algebras,
we obtain a new validity theorem, which
does not mention distributive laws at all, but is stated purely in
terms of causal algebras.
\begin{theorem}\label{thm:sol:causal}
  Let $B \colon \Set \rightarrow \Set$ be a polynomial functor, with
  final coalgebra $(B_\omega,\zeta)$.  Let
  $\alpha \colon FB_\omega \rightarrow B_\omega$ be a causal algebra.
  Then coinduction up to $\alpha$ is valid.
%  For every $f \colon X \rightarrow BFX$ there is a unique
%  $\sol{f} \colon X \rightarrow Z$ such that the following diagram
%  commutes.
%  $$
%  \xymatrix{
%    X \ar[d]_{f} \ar[rr]^{\sol{f}}
%    & & B_\omega \ar[d]^{\zeta} \\
%    BFX \ar[r]_-{BF\sol{f}}
%    & BFB_\omega \ar[r]_{B\alpha}
%    & BB_\omega
%  }
%  $$
\end{theorem}
\begin{proof}
  Consider the codensity monad $\cf{B}$, with counit $\epsilon$.
  By Lemma~\ref{lm:connected}, the functor $B$ satisfies the hypotheses of
  Theorem~\ref{thm:companion-codensity} and hence $\cf{B}$ is the underlying
  functor of the companion.
  By Lemma~\ref{lm:alg-codensity} (and since any polynomial functor is $\omega$-continuous), $\epsilon_\omega$ is the algebra
  induced by $\tau$ on the final coalgebra.
  Thus, by Corollary~\ref{cor:solutions-comp}, coinduction up to $\epsilon_\omega$ is valid.

  Since $\alpha$ is causal, by
  Corollary~\ref{cor:final-causal-alg} there is a unique natural
  transformation $\hat{\alpha} \colon F \Rightarrow \cf{B}$ such that
  $\epsilon_\omega \circ \hat{\alpha}_{B_\omega} = \alpha$.

  Let $f \colon X \rightarrow BFX$.
  Since coinduction up to $\epsilon_\omega$ is valid, there is
  a unique $\sol{f}$ making the outside of the following diagram commute (note that
  we abuse notation, using $\sol{f}$ to refer to the unique map associated
  to the coalgebra $B\hat{\alpha}_X \circ f$ by the validity of coinduction up to $\epsilon_\omega$).
  \[
  \xymatrix{
	  X \ar[d]_{f} \ar[rr]^{\sol{f}} % chktex 3
		  & & B_\omega \ar[d]^{\zeta} \\ % chktex 3
	  BFX \ar[r]^-{BF\sol{f}} \ar[d]_{B\hat{\alpha}_X} % chktex 3
		  & BFB_\omega \ar[r]^{B\alpha} \ar[d]_{B\hat{\alpha}_{B_\omega}} % chktex 3
		  & BB_\omega \ar@{=}[d]\\
	  B\cf{B}X \ar[r]_{B\cf{B}\sol{f}} % chktex 3
		  & B\cf{B}B_\omega \ar[r]_{B\epsilon_\omega} % chktex 3
		  & BB_\omega
  }
  \]
  The lower left square commutes by naturality, and the lower right square
  by definition of $\hat\alpha$. Thus the outside of the diagram
  commutes if and only if the inner rectangle commutes. It follows that
  $\sol{f}$ is the unique map making the rectangle commute, which is
  what we needed to prove.
%  Consider the codensity monad $\cf{B}$, with counit $\epsilon$. By
%  Lemma~\ref{lm:connected}, the functor $B$ satisfies the hypotheses of
%  Theorem~\ref{thm:companion-codensity}. By
%  Corollary~\ref{cor:final-causal-alg}, there is a unique natural
%  transformation $\hat{\alpha} \colon F \Rightarrow \cf{B}$ such that
%  $\epsilon_\omega \circ \hat{\alpha}_{B_\omega} = \alpha$.  By
%  Theorem~\ref{thm:companion-codensity}, there is a distributive law
%  $\tau$ of the monad $(\cf{B},\eta,\mu)$ over $\fin{B}$.  By
%  Lemma~\ref{lm:alg-codensity} (and since any polynomial functor is $\omega$-continuous), $\epsilon_\omega$ is the algebra
%  induced by $\tau$ on the final coalgebra.
%
%  Let $f \colon X \rightarrow BFX$. By
%  Corollary~\ref{cor:solutions-comp}, there exists a unique $\sol{f}$
%  making the outside of the following diagram commute (for $f$ in the
%  statement of the corollary we take $B\hat{\alpha}_X \circ f$).
%  $$
%  \xymatrix{
%	  X \ar[d]_{f} \ar[rr]^{\sol{f}}
%		  & & B_\omega \ar[d]^{\zeta} \\
%	  BFX \ar[r]^-{BF\sol{f}} \ar[d]_{B\hat{\alpha}_X}
%		  & BFB_\omega \ar[r]^{B\alpha} \ar[d]_{B\hat{\alpha}_{B_\omega}}
%		  & BB_\omega \ar@{=}[d]\\
%	  B\cf{B}X \ar[r]_{B\cf{B}\sol{f}}
%		  & B\cf{B}B_\omega \ar[r]_{B\epsilon_\omega}
%		  & BB_\omega
%  }
%  $$
%  The lower left square commutes by naturality, and the lower right square
%  by definition of $\hat\alpha$. Thus the outside of the diagram
%  commutes if and only if the inner rectangle commutes. It follows that
%  $\sol{f}$ is the unique map making the rectangle commute, which is
%  what we needed to prove.
\end{proof}
\begin{example}%
  \label{ex:stream-ops}
  For the functor $BX = \real \times X$, $B_\omega$ is the set of streams.
  Let $SX=X^2$, and consider the coalgebra
  $f \colon 1 \rightarrow BS1$ with $1=\{*\}$, defined by
  $* \mapsto (1,(*,*))$.  Pointwise addition is a causal function on
  streams, modelled by an algebra on $B_\omega$ for the functor
  $S$. By Theorem~\ref{thm:sol:causal} we obtain a unique solution
  $\sigma \in B_\omega$, satisfying $\sigma_0 = 1$ and
  $\sigma' = \sigma \oplus \sigma$. Similarly, the shuffle product of
  streams is causal, so that by applying Theorem~\ref{thm:sol:causal}
  with that algebra and the same coalgebra $f$ we obtain a unique
  stream $\sigma$ satisfying $\sigma_0 = 1$,
  $\sigma' = \sigma \otimes \sigma$.
  %The latter, in fact, is the
%  sequence of Catalan numbers, see, e.g.,~\cite{HansenKR16}.

  As explained in the Introduction, this method also allows one to
  define functions on streams. For instance, for the shuffle product,
  define a $BS$-coalgebra
  $f \colon {(B_\omega)}^2 \rightarrow {BS(B_\omega)}^2$ by
  $f(\sigma,\tau) = (\sigma_0 \times \tau_0,
  ((\sigma',\tau),(\tau,\sigma')))$.
  Since addition of streams is causal, by Theorem~\ref{thm:sol:causal}
  there is a unique
  $\sol{f} \colon B_\omega \times B_\omega \rightarrow B_\omega$ such
  that $\sol{f}(\sigma,\tau)(0) = \sigma(0) \times \tau(0)$ and
  $(\sol{f}(\sigma,\tau))' = (\sol{f}(\sigma', \tau) \oplus
  \sol{f}(\sigma,\tau'))$,
  matching the definition given in the
  Introduction~(\ref{eq:times1}). Notice that not every function
  defined in this way is causal; for instance, it is easy to define
  $\even$ (see Example~\ref{ex:causal-streams}), even with the
  standard coinduction principle (i.e., where $F = \Id$ and
  $\alpha=\id$).
\end{example}

\begin{example}%
  \label{ex:cfg}
  Consider the functor $BX = 2 \times X^A$, whose final coalgebra
  consists of the set $\pow(A^*)$ of languages. A $B\pow$-coalgebra
  $f \colon X \rightarrow 2 \times {(\pow(X))}^A$ is a non-deterministic
  automaton. Taking the causal algebra
  $\alpha \colon \pow(\pow(A^*)) \rightarrow \pow(A^*)$ defined by
  union, the unique map $\sol{f} \colon X \rightarrow \pow(A^*)$ from
  Theorem~\ref{thm:sol:causal} is the usual language semantics of
  non-deterministic automata.

  In~\cite{WinterBR13}, a context-free grammar (in Greibach normal
  form) is modelled as a $B\pow^*$-coalgebra
  $f \colon X \rightarrow 2 \times {({\pow(X)}^*)}^A$, where $X$ are the non-terminals,
  and its semantics is defined operationally by turning $f$ into a deterministic
  automaton with state space $\pow(X^*)$.  In~\cite{RW13} this operational view is
  related to the semantics of CFGs in terms of language equations.
  Consider the causal algebra
  $\alpha \colon \pow({\pow(A^*)}^*) \rightarrow \pow(A^*)$ defined by
  union and language composition:
  $\alpha(S) = \bigcup_{L_1 \ldots L_k \in S} L_1 L_2 \ldots L_k$.
  By Theorem~\ref{thm:sol:causal},
  any context-free grammar $f$ has a unique solution
  $\sol{f} \colon X \rightarrow \pow(A^*)$ assigning a language
  to every non-terminal; the commutativity of the diagram in~\ref{thm:sol:causal}
  amounts to the fact that this is a solution of the grammar $f$ viewed
  as a system of equations over the set $\pow(A^*)$ of languages.
  As such, we obtain an elementary coalgebraic semantics of
  CFGs that does not require us to relate it to an
  operational semantics.
  %	A context-free grammar is in Greibach normal form if all
  % productions from a non-terminal $S$ are either of the form
  % $S \rightarrow \varepsilon$, where $\varepsilon$ is the empty
  % word, or of the form $S \rightarrow a S_1S_2 \ldots S_n$ where $a$
  % is an alphabet letter, and $S_1, \ldots, S_n$ are again
  % non-terminals. Such a grammar can be interpreted as a system of
  % equations on languages; its unique solution is the semantics of
  % the grammar.  In~\cite{}, a context-free grammar (in Greibach
  % normal form) is modelled as a coalgebra
  % $f \colon X \rightarrow 2 \times (\pow(X)^*)^A$, where $X$ is the
  % set of non-terminals, and semantics is defined operationally by
  % turning $f$ into a deterministic automaton over $\pow(X)^*$.
  % In~\cite{}, this operational view was proved equivalent to the
  % semantics in terms of language equations.  In the current setting,
  % we obtain the coalgebraic semantics of context-free grammars in
  % terms of language equations. To this end, take the algebra
  % $\alpha \colon \pow(\pow(A^*)^*) \rightarrow \pow(A^*)$, defined
  % by mapping a set of lists of languages to the union of their
  % composition.  This algebra is clearly causal, hence by
  % Theorem~\ref{} we obtain the coalgebraic semantics of context-free
  % grammars, without having to explicitly define an operational
  % semantics.
\end{example}

\subsection{Causal algebras and distributive laws}\label{sec:causal-alg-dl}

Another application of the fact that the codensity monad is the
companion is that the final causal algebra in
Corollary~\ref{cor:final-causal-alg} is, by
Lemma~\ref{lm:alg-codensity}, the algebra induced by a
distributive law. Hence, \emph{any} causal algebra is ``definable'' by
a distributive law, in the sense that it factors as a (component of a)
natural transformation followed by the algebra induced by a
distributive law. This is stated in Corollary~\ref{cor:causal-iff-dl} below, which follows from the following more general result.
\begin{theorem}\label{thm:definable}
	Let $A,B \colon \Set \rightarrow \Set$ be functors, where $A$ is $\omega$-continuous and $B$ is polynomial.
	An algebra $\alpha_\omega \colon FA_\omega \rightarrow B_\omega$
	is causal if and only if there is a functor $G$, a natural transformation $\lambda \colon GA \Rightarrow BG$
	and a natural transformation $\kappa \colon F \Rightarrow G$ such that
	the following diagram commutes:
	\[
	\xymatrix{
	  FA_\omega \ar[rr]^{\kappa_{A_\omega}} \ar[dr]_{\alpha_\omega} & & GA_\omega \ar[dl]^{\beta_\omega} \\ % chktex 3
	  & B_\omega &
	}
	\]
	where $\beta_\omega$ is the unique map induced by $\lambda$ as in~\eqref{eq:unique-map}.
	% such that
%	$$
%	  \xymatrix{
%		  FA_k \ar[r]^{\beta_k}\ar[d]_{FA_{k+1,k}^{-1}}
%			  & B_k \ar[dd]^{B_{k+1,k}^{-1}} \\
%		  FAA_k \ar[d]_{\lambda_{A_k}}
%			  & \\
%		  BFA_k \ar[r]_{B\beta_k}
%			  & BB_k
%	  }
%	$$
\end{theorem}

\begin{proof}
	First note that the type of $\beta_\omega$ is correct: since $A,B$ are both $\omega$-continuous,
	$A_\omega$ and $B_\omega$ are final coalgebras.
	For the implication from right to left, by Lemma~\ref{lm:dl-alpha}, a natural transformation
	$\lambda \colon GA \Rightarrow BG$ defines a natural transformation $\beta \colon G\fin{A} \Rightarrow \fin{B}$
	such that $\beta_\omega$ is the unique map induced by $\lambda$, and
	hence satisfies the above diagram.
	By Theorem~\ref{thm:causal-alg},
	$\alpha_\omega = \beta_\omega \circ \kappa_{A_\omega}$ is causal, since
	$\beta \circ \kappa \fin{A} \colon F\fin{A} \Rightarrow \fin{B}$
	is a natural transformation.

	For the converse, let $\alpha \colon FA_\omega \rightarrow B_\omega$ be causal.
	We instantiate $G$ and $\beta_\omega$ in the statement of the theorem respectively to $\kfab$
	and $\epsilon_\omega$, for $\epsilon$ its counit.
	By Corollary~\ref{cor:final-causal-alg}, there is a natural transformation $\hat{\alpha} \colon F \Rightarrow \kfab$
	such that $\alpha = \epsilon_{\omega} \circ \hat{\alpha}_{A_\omega}$. %, where $\epsilon$ is the counit of $\kfab$.
	Further, by Lemma~\ref{lm:connected}, $B$ satisfies the hypotheses of Theorem~\ref{thm:companion-codensity}.
	Hence, $\kfab$ is the companion with a natural transformation $\tau \colon \kfab A \Rightarrow B\kfab$, and
	by Lemma~\ref{lm:alg-codensity}, $\epsilon_\omega$ is the unique map induced by $\tau$.
\end{proof}
In particular, taking $A=B$, the above theorem is an expressivity result
for algebras defined by distributive laws. To make this more precise,
suppose $B \colon \Set \rightarrow \Set$ has a final
coalgebra $(Z,\zeta)$.  We say an algebra
$\alpha \colon FZ \rightarrow Z$ is \emph{definable by a distributive
  law} if there exists a distributive law
$\lambda \colon GB \Rightarrow BG$ with induced algebra
$\beta \colon GZ \rightarrow Z$ and a natural transformation
$\kappa \colon F \Rightarrow G$ such that the following commutes:
\[
\xymatrix{
  FZ \ar[rr]^{\kappa_{Z}} \ar[dr]_{\alpha} & & GZ \ar[dl]^{\beta} \\ % chktex 3
  & Z &
}
\]
\begin{corollary}\label{cor:causal-iff-dl}
  Let $B \colon \Set \rightarrow \Set$ be polynomial.  An algebra
  $\alpha_\omega \colon FB_\omega \rightarrow B_\omega$ is causal if and only
  if it is definable by a distributive law.
\end{corollary}
\begin{example}
Since the functors for stream systems and automata are polynomial, by
Corollary~\ref{cor:causal-iff-dl} we obtain that a
function $f \colon {(\real^\omega)}^V \rightarrow \real^\omega$ on streams over $\real$,
or a function $f \colon {(\pow(A^*))}^V \rightarrow \pow(A^*)$ on languages over $A$,
is causal if and only if
it is definable by a distributive law.

As explained in Example~\ref{ex:causal-even}, the function $\even$ is
not causal (from $B$ to $B$, where $B = \real \times \Id$) but it is causal from $A$ to $B$,
where $A = \real \times \real \times \Id$. It follows from Corollary~\ref{cor:causal-iff-dl} %Theorem~\ref{thm:definable}
that $\even$ is \emph{not} definable by a distributive law of the form $GB \Rightarrow BG$.
However, by Theorem~\ref{thm:definable}, there is a distributive law
$\lambda \colon GA \Rightarrow BG$ inducing an algebra $\beta_\omega \colon G\real^\omega \rightarrow \real^\omega$
(recall that $\real^\omega$ is the final coalgebra of both $A$ and $B$) and
a natural transformation $\kappa \colon \Id \Rightarrow G$ such that $\even = \beta_\omega \circ \kappa_{\real^\omega}$.
Indeed, $\even$ arises directly as the unique operation induced by a natural transformation of the form $GA \Rightarrow BG$
where $G = \Id$,
as shown in Example~\ref{ex:even-gen-dl}.
\end{example}

In~\cite{HansenKR16}, a similar result to
Corollary~\ref{cor:causal-iff-dl} is shown concretely for causal
stream functions, and this is extended to languages
in~\cite{RotBR16}. In both cases, very specific presentations of
distributive laws for the systems at hand are used to present the
distributive law based on a ``syntax'', which however is not too
clearly distinguished from the semantics: it consists of a single
operation symbol for every causal function.  In our case, in the proof
of Theorem~\ref{thm:definable}, we use the \emph{companion}, which
consists of the actual functions rather than a syntactic
representation. Indeed, the setting of Theorem~\ref{thm:definable}
applies more abstractly to \emph{all} causal algebras, not just causal
functions. However, it remains an intriguing question how to obtain a
concrete syntactic characterisation of a distributive law for a given
causal algebra.

%Our result applies immediately to deterministic automata, since the functor $X \mapsto 2 \times X^A$ is polynomial as well.
%We thus immediately obtain that function on languages is causal iff it is definable by a distributive law,
%recovering the concrete result in~\cite{RotBR16} (again, modulo a syntactic presentation of causal functions). There are many other examples of polynomial functors whose coalgebras model systems of interest (see~\cite{Jacobs:coalg}), such as $X \mapsto O \times X^A$ for any set $O$, whose coalgebras are Moore machines, $X \mapsto (A \times X)^I$, whose coalgebras are Mealy machines, or $X \mapsto A \times X^*$, whose final coalgebra consists of infinite trees node-labelled in $A$.

\subsection{Soundness of up-to techniques}%
\label{ssec:fibrations}

The \emph{contextual closure} of an algebra is one of the most
powerful up-to techniques, which allows one to exploit algebraic
structure in bisimulation proofs. In~\cite{bppr:acta:16}, it is shown
that the contextual closure is sound (compatible) on any bialgebra for
a distributive law. Here, we move away from the explicit requirement of a distributive law
and give an elementary condition for soundness of the contextual closure on the
final coalgebra: that the algebra under consideration is causal. In
fact, we prove that this implies that the contextual closure lies
below the companion, which not only gives soundness, but also allows
to combine it with other up-to techniques.

%ARDEN'S RULE EXAMPLE
\begin{example}\label{ex:up-to-da}
To motivate and illustrate the use of contextual closure as an up-to technique, we recall
from~\cite{RotBR16} a coinductive proof of \emph{Arden's rule}.

First let $\Rel_{\pow(A^*)} = \pow(\pow(A^*) \times \pow(A^*))$, and
consider the function $b \colon \Rel_{\pow(A^*)} \rightarrow \Rel_{\pow(A^*)}$ defined by
\[b(R) = \{(L,K) \mid \varepsilon \in L \text{ iff } \varepsilon \in K, \text{ and } \forall a \in A. \, (L_a, K_a)\} \]
where, given $L \in \pow(A^*)$ and $a \in A$, $L_a = \{w \mid aw \in L\}$, called language derivative.
A relation $R$ is a \emph{bisimulation} if $R \subseteq b(R)$; concretely, $R$ is a bisimulation if for every pair $(L,K) \in R$:
$L$ contains the empty word iff $K$ does, and $(L_a,K_a) \in R$ for all $a \in A$. The greatest fixpoint of $b$ is language equality.
Hence, the coinduction principle asserts that languages $L,K$ are equal whenever they are contained in a bisimulation.

Arden's rule states that for every three languages
$L,K,M \in \pow(A^*)$, if $L=KL + M$ and $K$ does not contain the empty word, then $L=K^*M$. To prove it, one may try to show that, for $L,K,M$ satisfying
the assumption, $R = \{(L,K^*M)\}$ is a bisimulation. However, this fails, since
\[
L_a = {(KL+M)}_a = K_{a}L + M_a, \qquad \text{whereas} \qquad {(K^*M)}_a = K_{a}K^*M + M_a\,,
\]
using so-called Brzozowski derivatives to compute the $a$-dervatives~\cite{Rutten98}. The pair $(L_a,{(K^*M)}_a)$
is not related by $R$. However,
it is related by the bigger relation $\ctx(\rfl(R))$, where $\rfl, \ctx \colon \Rel_{\pow(A^*)} \rightarrow \Rel_{\pow(A^*)}$
are the following functions on relations:
$\rfl(R) = R \cup \{(L,L) \mid L \subseteq A^*\}$ is the reflexive closure and
$\ctx(R)$ is the \emph{contextual closure} of $R$ (with respect to sum and composition).
More precisely, $\ctx(R)$ is
the least relation satisfying the following rules:
\begin{align*}
&\frac{
(L,K) \in R}{
(L,K) \in \ctx(R)
}
\quad
\frac{
(L_1,K_1) \in \ctx(R) \quad
(L_2,K_2) \in \ctx(R)}
{(L_1 + L_2,K_1 + K_2) \in \ctx(R)}
\\
&\qquad\qquad\frac{
(L_1,K_1) \in \ctx(R) \quad
(L_2,K_2) \in \ctx(R)}
{(L_1 L_2,K_1 K_2) \in \ctx(R)}
\end{align*}
Hence we have $R \subseteq b(\ctx(\rfl(R)))$,
which means that $R$ is a \emph{bisimulation up to $\ctx \circ \rfl$}. To conclude that $L=K^*M$,
it suffices to prove that $\ctx$ and $\rfl$ are both below the companion of $b$ (cf.~\cite{pous:lics16:cawu}).
Both $\rfl$ and $\ctx$ are, in fact, compatible, which follows from~\cite{bppr:acta:16}.
We focus on $\ctx$. Showing compatibility using the techniques of~\cite{bppr:acta:16} requires providing a distributive law.
However, it also follows as an instance of Theorem~\ref{thm:causal-below-t} below
that $\ctx$ is below the companion, relying on causality
instead of having to provide an explicit distributive law.
\end{example}

We generalise the functions $b$ and $\ctx$ from the above examples
to speak more abstractly about soundness of the contextual closure
for bisimilarity proofs, following the approach of~\cite{bppr:acta:16} to up-to techniques.
This approach is formulated at the abstract level of coinductive predicates in fibrations.
However, we only recall a few necessary definitions, and refer to~\cite{bppr:acta:16}
for details. %A full treatment is beyond the scope of this paper.

For a set $X$, let $\Rel_X = \pow(X \times X)$ be the lattice of relations, ordered
by subset inclusion.
For a functor $B \colon \Set \rightarrow \Set$,
bisimulations on a $B$-coalgebra $(X,f)$ are the post-fixed points of
the monotone function $b_f \colon \Rel_X \rightarrow \Rel_X$, defined by
\[b_f(R) = {(f\times f)}^{-1} \circ \Rel(B)(R) \,.\]
Here ${(f \times f)}^{-1}$ is inverse image along $f\times f$, and $\Rel(B)$ is the
\emph{relation lifting} of $B$, defined for any relation $R$ with
projections $\pi_1, \pi_2$ by
  \[
  \Rel(B)(R) = \{(x,y) \mid \exists z \in BR . \, x = B\pi_1(z), \, y = B\pi_2(z)\}\,,
  \]
 see, e.g.,~\cite{Jacobs:coalg}.
Contextual closure
$\ctx_\alpha \colon \Rel_X \rightarrow \Rel_X$ with respect to an
algebra $\alpha \colon FX \rightarrow X$ is defined dually by
\[\ctx_\alpha(R) = \textstyle{\coprod}_\alpha \circ \Rel(F)(R)\]
where
$\textstyle{\coprod}_\alpha$ is direct image along
$\alpha \times \alpha$.

We first prove a general property of algebras and the contextual closure.
\begin{lemma}\label{lm:ctx-alg}
  Let $X$ be a set, $F, G \colon \Set \rightarrow \Set$ functors,
  $\alpha \colon FX \rightarrow X$ and $\beta \colon GX \rightarrow X$
  algebras, and $\kappa \colon F \Rightarrow G$ a natural transformation, such that $\alpha = \beta \circ
  \kappa_X$. %as below:
  % $$
  % \xymatrix{
  %	FX \ar[rr]^{\kappa_{X}} \ar[dr]_{\alpha} & & GX \ar[dl]^{\beta} \\
  % & X &
  % }
  %   $$
  Then $\ctx_\alpha \leq \ctx_\beta$.
  %	Consequently, for any $b \colon \Rel_X \rightarrow \Rel_X$
  % with companion $t$, if $\ctx_\beta \leq t$
  %	then $\ctx_\alpha \leq t$.
\end{lemma}
\begin{proof}
  %	This proof relies on the setting and terminology
  % of~\cite{bppr:acta:16}, which we do not fully recall here.
  The natural transformation $\kappa$ lifts to a natural
  transformation \[\Rel(\kappa) \colon \Rel(F) \Rightarrow \Rel(G),\]
  see~\cite[Exercise 4.4.6]{Jacobs:coalg}. It follows from a general
  property of fibrations (see~\cite[Lemma 14.5]{bppr:acta:16}) that
  there exists a natural transformation of the form
  \[\textstyle{\coprod}_{\kappa_X} \circ \Rel(F) \Rightarrow \Rel(G)
  \colon \Rel_X \rightarrow \Rel_{GX}.\]  Hence, we obtain a natural
  transformation
  \begin{align*}
    \ctx_\alpha &= \textstyle{\coprod}_{\alpha} \circ \Rel(F) \\
                & = \textstyle{\coprod}_{\beta \circ \kappa_X} \circ \Rel(F) \\
                & = \textstyle{\coprod}_{\beta} \circ \textstyle{\coprod}_{\kappa_X} \circ \Rel(F) \\
                &\Rightarrow \textstyle{\coprod}_{\beta} \circ \Rel(G) = \ctx_{\beta} \,.
  \end{align*}
  This is a natural transformation in $\Rel_X$, which just means that
  $\ctx_\alpha \leq \ctx_\beta$.
\end{proof}

\begin{theorem}\label{thm:causal-below-t}
  Let $B \colon \Set \rightarrow \Set$ be a polynomial functor, and
  $(B_\omega, \zeta)$ a final $B$-coalgebra.  Let $t_\zeta$ be the
  companion of $b_\zeta$. For any causal algebra
  $\alpha \colon FB_\omega \rightarrow B_\omega$, we have
  $\ctx_\alpha \leq t_\zeta$.
\end{theorem}
\begin{proof} %[of Theorem~\ref{thm:causal-below-t}]
  By Lemma~\ref{lm:connected}, $B$ satisfies the hypotheses of
  Theorem~\ref{thm:companion-codensity}, and hence by
  Lemma~\ref{lm:alg-codensity}, $\epsilon_\omega$ is the
  algebra induced by the distributive law $\tau$ of the companion.
  This means that $(B_\omega,\epsilon_\omega,\zeta)$ is a
  $\tau$-bialgebra, and it follows from~\cite[Corollary
  6.8]{bppr:acta:16} that $\ctx_{\epsilon_\omega}$ is
  $b_\zeta$-compatible.  Thus
  $\ctx_{\epsilon_\omega} \leq t_\zeta$.
  Now, let $\alpha \colon FB_\omega \rightarrow B_\omega$ be causal.  By
  Corollary~\ref{cor:final-causal-alg}, there exists a natural
  transformation $\hat{\alpha} \colon F \Rightarrow \cf{B}$ such that
  $\alpha = \epsilon_{\omega} \circ \hat{\alpha}_{B_\omega}$.  By
  Lemma~\ref{lm:ctx-alg} we obtain
  $\ctx_\alpha \leq \ctx_{\epsilon_\omega}$, hence
  $\ctx_\alpha \leq t_\zeta$.
\end{proof}

This implies that one can safely use the contextual closure for
\emph{any} causal algebra, such as union, concatenation and Kleene
star of languages, or product and sum of streams.
In particular, we recover the soundness of the contextual closure
in Example~\ref{ex:up-to-da} from the above theorem
and the simple observation that language union and composition are causal.

Endrullis et
al.~\cite{EndrullisHB13} prove the soundness of \emph{causal contexts}
in combination with other up-to techniques, for equality of
streams. The soundness of causal algebras for streams is a special
case of Theorem~\ref{thm:causal-below-t}, but the latter provides
more: being below the companion, it is possible to compose it to other
such functions to obtain combined up-to techniques in a modular
fashion, cf.~\cite{pous:lics16:cawu}.

%	Given an algebra $\alpha \colon FX \rightarrow X$, we define
%	$$
%	\ctx_\alpha = {\textstyle{\coprod}_{\alpha}} \circ {\Rel(F)}  : \Rel_X \rightarrow \Rel_X \,.
%	$$
%\begin{lemma}
%	Let $X$ be a set, and consider algebras $\alpha,\beta$ and a natural transformation $\kappa$ as below:
%$$
%\xymatrix{
%	FX \ar[rr]^{\kappa_{X}} \ar[dr]_{\alpha} & & GX \ar[dl]^{\beta} \\
%		& X &
%}
%$$
%	Then $\ctx_\alpha \leq \ctx_\beta$.
%	Consequently, for any $b \colon \Rel_X \rightarrow \Rel_X$ with companion $t$, if $\ctx_\beta \leq t$
%	then $\ctx_\alpha \leq t$.
%\end{lemma}
%
%\marginpar{lemma to appendix, corollary becomes a theorem}
%
%\begin{corollary}
%	Let $B \colon \Set \rightarrow \Set$ be a functor which preserves connected limits,
%	and let $\zeta \colon B_\omega \rightarrow BB_\omega$ be a final coalgebra.
%	Let $t$ be the companion of $\zeta^* \circ \Rel(B)$.
%	For any causal algebra $\alpha \colon FB_\omega \rightarrow B_\omega$:
%	$\ctx_\alpha \leq t$.
%\end{corollary}

\section{Abstract GSOS}%
\label{sec:gsos}

To obtain expressive specification formats, Turi and
Plotkin~\cite{TuriP97} use natural transformations of the form
$\lambda \colon F(B \times \Id) \Rightarrow BF^*$, where $F^*$ is the
free monad for $F$. These are the so-called \emph{abstract GSOS
  specifications}. In this section we show that they are
actually equally expressive as plain distributive laws of a functor
$F$ over $B$, if the conditions of Theorem~\ref{thm:companion-codensity} apply
(assuring that the codensity monad is the companion).
This is in a similar spirit as Section~\ref{sec:causal-alg-dl}, but we give a proof
here that does not require results on causal algebras.

If $B$ has a final coalgebra $(Z,\zeta)$, then any abstract GSOS
specification $\lambda \colon F(B \times \Id) \Rightarrow BF^*$
defines an algebra $\alpha \colon FZ \rightarrow Z$ on it, which is
the unique algebra making the following diagram commute.
\[
\xymatrix{
	FZ \ar[d]_{\alpha} \ar[r]^-{F\langle \zeta, \id \rangle} % chktex 3
		& F(B \times \Id)Z \ar[r]^-{\lambda_Z} % chktex 3
		&  BF^*Z \ar[d]^{B\alpha^*} \\ % chktex 3
	Z \ar[rr]_{\zeta} % chktex 3
		& & BZ
}
\]
Here $\alpha^*$ is the Eilenberg-Moore algebra for the free monad
corresponding to $\alpha$.  Intuitively, this algebra gives the
interpretation of the operations defined by $\lambda$.

Like plain distributive laws (Lemma~\ref{lm:dl-alpha}), abstract GSOS
specifications induce natural transformations of the form
$F\fin{B} \Rightarrow \fin{B}$.
\begin{lemma}\label{lm:dl-alpha-gsos}
  For every $\lambda \colon F(B \times \Id) \Rightarrow BF^*$ there is
  a unique $\alpha \colon F\fin{B} \Rightarrow \fin{B}$ such that for
  all $i \in \Ord$:
  $\alpha_{i+1} = B\alpha_i^* \circ \lambda_{B_i} \circ F\langle \id,
  B_{i+1,i}\rangle$.
  Moreover, if $B_{k+1,k}$ is an isomorphism for some $k$, then
  $\alpha_k$ is the algebra induced by $\lambda$ on the final
  coalgebra.
\end{lemma}
\begin{proof}
  The transformation $\alpha$ is determined by the successor case
  given in the definition.  Naturality is proved in a similar way as
  in Lemma~\ref{lm:dl-alpha}, with the relevant diagram in the
  successor case replaced by:
  \[
  \xymatrix@C=1.5cm{
    FBB_j \ar[r]^-{F\langle \id, B_{j+1,j}\rangle} \ar[d]_{FBB_{j,i}} % chktex 3
    & F(B\times \Id) B_j \ar[r]^-{\lambda_{B_j}} \ar[d]^{F(B \times \Id)B_{j,i}} % chktex 3
    & BF^*B_j \ar[r]^{B\alpha_j^*} \ar[d]^{BF^*B_{j,i}} % chktex 3
    & BB_j \ar[d]^{BB_{j,i}}\\ % chktex 3
    FBB_i \ar[r]_-{F\langle \id, B_{i+1,i}\rangle} % chktex 3
    & F(B\times \Id) B_i \ar[r]_-{\lambda_{B_i}} % chktex 3
    & BF^*B_i \ar[r]_{B\alpha_i^*} % chktex 3
    & BB_i
  }
  \]
  The left square commutes since
  $B_{i+1,i} \circ BB_{j,i} = B_{i+1,i} \circ B_{j+1,i+1} = B_{j+1,i}
  = B_{j,i} \circ B_{j+1,j}$
  by functoriality and definition of the final sequence. The middle
  square commutes by naturality.  The one on the right commutes, since
  $B_{j,i}$ is (by assumption in the inductive proof) an algebra
  morphism, i.e., $B_{j,i} \circ \alpha_j = \alpha_i \circ FB_{j,i}$,
  and hence $B_{j,i} \circ \alpha_j^* = \alpha_i^* \circ F^*B_{j,i}$
  (it holds in general that the ${(-)}^*$ construction preserves algebra
  homomorphisms).

  Suppose $B_{k+1,k} \colon B_{k+1} \rightarrow B_{k}$ is an
  isomorphism. Then $B_{k+1,k}^{-1} \colon B_k \rightarrow B(B_{k+1})$
  is a final $B$-coalgebra. Consider the following diagram:
  \[
  \xymatrix@C=1.5cm{
    FB_k \ar[dd]_{\alpha_k} \ar[r]^-{F\langle B_{k+1,k}^{-1}, \id \rangle} \ar[dr]_{FB_{k+1,k}^{-1}} % chktex 3
    & F(B \times \Id)B_k \ar[r]^-{\lambda_{B_k}} % chktex 3
    &  BF^*B_k \ar[dd]^{B\alpha_k^*} \\ % chktex 3
    &
    FBB_k \ar[u]_{F\langle \id, B_{k+1,k}\rangle} \ar[dr]^{\alpha_{k+1}} % chktex 3
    & \\
    B_k \ar[rr]_{B_{k+1,k}^{-1}} % chktex 3
    & & BB_k
  }
  \]
  The big triangle commutes by naturality and the fact that
  $B_{k+1,k}$ is an isomorphism, the small triangle since $B_{k+1,k}$
  is an isomorphism, and the remaining inner shape by definition of
  $\alpha$. Hence, $\alpha_k$ is the algebra induced on the final
  coalgebra by $\lambda$.
\end{proof}
This places abstract GSOS specifications within the framework of the
companion, constructed via the codensity monad of the final sequence
$\fin{B}$. Whenever that construction applies (e.g., for polynomial
functors), any algebra defined by an abstract GSOS is thus already
definable by a plain distributive law over $B$.
\begin{theorem}\label{thm:abs-gsos-def}
  Let $B \colon \C \rightarrow \C$ such that $\cf B$ exists, $B$
  preserves $\cf B$, and $B_{k+1,k} \colon B_{k+1} \rightarrow B_k$ is an
  iso for some $k$.  Every algebra induced on the final coalgebra by
  an abstract GSOS specification
  $\lambda \colon F(B \times \Id) \Rightarrow BF^*$ is definable by a
  distributive law over $B$ (cf. Section~\ref{sec:causal-alg-dl}).
\end{theorem}
\begin{proof}
  By Lemma~\ref{lm:dl-alpha-gsos}, the algebra induced by an abstract
  GSOS $\lambda$ is given by $\alpha_k$ for some
  $\alpha \colon F\fin{B} \Rightarrow \fin{B}$.  By the universal
  property of the codensity monad $(\cf{B}, \epsilon)$, there exists a
  (unique) natural transformation
  $\hat{\alpha} \colon F \Rightarrow \cf{B}$ such that
  $\alpha = \epsilon \circ \hat{\alpha} \fin{B}$.  This means in particular
  that $\alpha_k = \epsilon_k \circ \hat{\alpha}_{B_k}$. By
  Lemma~\ref{lm:alg-codensity}, $\epsilon_k$ is the
  algebra induced by a distributive law (i.e., the companion), so $\alpha_k$ is definable by
  a distributive law over $B$.
\end{proof}
In this sense, abstract GSOS is no more expressive than plain
distributive laws. Note, however, that this does involve moving to a
different (larger) syntax.

% Theorem~\ref{thm:companion-codensity} applies, for instance, if $B$ is
% a polynomial functor on $\Set$ (Lemma~\ref{lm:connected}).

\begin{remark}
  Every abstract GSOS specification
  $\lambda \colon F(B \times \Id) \Rightarrow BF^*$ corresponds to a
  unique distributive law
  $\lambda^\dagger \colon F^*(B \times \Id) \Rightarrow (B \times \Id)F^*$
  of the free monad $F^*$ over the (cofree) copointed functor
  $B \times \Id$, see~\cite{LenisaPW04}. The algebra induced by
  $\lambda$ decomposes as the algebra induced by $\lambda^\dagger$ and
  the canonical natural transformation $F \Rightarrow F^*$.  This
  implies that every algebra induced by an abstract GSOS is definable
  by a distributive law over the copointed functor $B \times
  \Id$.
  Theorem~\ref{thm:abs-gsos-def} strengthens this to definability by a
  distributive law over $B$.
\end{remark}

% Indeed, in this case, Lemma~\ref{lm:dl-alpha-gsos} and
% Corollary~\ref{cor:final-causal-alg} assert that any such GSOS
% defines only causal operations.

\section{Preserving the right Kan extension}%
\label{sec:pres-pow}

Our main result for constructing the companion of functors $(A,B)$
requires that $B$ preserves the right Kan extension
$\Ran{\fin{A}}{\fin{B}}$ (Theorem~\ref{thm:companion-codensity}).  In
Section~\ref{sec:comp-set} we focused on polynomial functors on
$\Set$, which always satisfy this condition.
However, polynomial functors exclude a particularly important example:
the (finite) powerset functor $\powf$. This functor and its variants
are used to model, for instance, labelled transition systems as coalgebras.
In the current section we distill a concrete condition on a functor
$B \colon \Set \rightarrow \Set$ to preserve the Kan extension as above,
taking $A=B$ to make the notation somewhat lighter. In particular, we show
a negative result: the finite powerset functor does \emph{not}
preserve the relevant right Kan extension, and hence falls outside
the scope of Theorem~\ref{thm:companion-codensity}. Notice that
$\powf$ could still have a companion, and it could even be the codensity
monad of the final sequence; however, it does not follow
from our results, and remains an open question.

It will be useful to slightly reformulate preservation.
Let $F \colon \C \rightarrow \D$, $G \colon \C \rightarrow \E$  and
$K \colon \E \rightarrow \F$ be functors.
Consider the right Kan extensions
$(\Ran{F}{G},\epsilon)$ and $(\Ran{F}{KG},\epsilon')$.
By the univeral property of the latter, there is a unique
\[
\widehat{K\epsilon} \colon K\Ran{F}{G} \Rightarrow \Ran{F}{KG}
\]
such that $\epsilon' \circ \widehat{K\epsilon} F = K\epsilon$.
Now, $K$ preserves $\Ran{F}{G}$ iff this canonical map $\widehat{K\epsilon}$
is an isomorphism.
\begin{lemma}\label{lm:canonical-map}
  In the above, if $\D=\E=\F = \Set$ and the right Kan extensions are
  presented as in Lemma~\ref{lm:codensity-set}, then the canonical map
  $\widehat{K\epsilon}$ is given by
  \[
  \begin{array}{rcccl}
    {({(\widehat{K\epsilon})}_X(S))}_A
    &\colon & {(FA)}^X &\rightarrow &KGA \\
    & & f  & \mapsto & K(\lambda \alpha. \alpha_A(f))(S)
  \end{array}
  \]
  for all sets $X,A$ and all $S \in K\Ran{F}{G}(X)$.
\end{lemma}
\begin{proof}
  By Lemma~\ref{lm:codensity-set}, we have that
  ${({(\widehat{K\epsilon})}_X(S))}_A \colon {(FA)}^X \rightarrow KGA$ is
  given by
  \[
  f \mapsto {(K\epsilon)}_A \circ K\Ran{F}{G}(f)(S) = K(\epsilon_A \circ
  \Ran{F}{G}(f))(S)\,.
  \]
  But, for all $\alpha \colon {(F-)}^X \Rightarrow G$, we have
  \begin{align*}
  \epsilon_A(\Ran{F}{G}(f)(\alpha))
  &= {(\Ran{F}{G}(f)(\alpha))}_A(\id_{FA}) \\
  &= (\lambda g. \alpha_A(g \circ f))(\id_{FA}) \\
  &= \alpha_A(\id_{FA} \circ f) \\
  &= \alpha_A(f)
  \end{align*}
%  \epsilon_A(\Ran{F}{G}(f))(\alpha) = \epsilon_A(\lambda g. \alpha_A(g
%  \circ f)) = \alpha_A(\id_{FA} \circ f) = \alpha_A(f)
  where the first two steps follow from Lemma~\ref{lm:codensity-set}.
\end{proof}
\noindent To simplify the notation, below we denote the finite
powerset by $\pfs$.

\begin{proposition}
  The finite powerset functor $\pfs \colon \Set \rightarrow \Set$ does
  not preserve $\Ran{\fin{\pfs}}{\fin{\pfs}}$.
\end{proposition}
\begin{proof}
  We show that the canonical map
  $\widehat{\pfs \epsilon} \colon \pfs \Ran{\fin{\pfs}}{\fin{\pfs}}
  \Rightarrow \Ran{\fin{\pfs}}{\pfs\fin{\pfs}}$ is not an isomorphism,
  using the characterisation in Lemma~\ref{lm:canonical-map}.  The
  latter gives
  \begin{equation}\label{eq:hatepspow}
    \begin{array}{rcccl}
      {({(\widehat{\pfs\epsilon})}_X(S))}_i
      &\colon & {(\pfs_i)}^X &\rightarrow &\pfs \pfs_i \\
      & & f  & \mapsto & \{\alpha_i(f) \mid \alpha \in S\}
    \end{array}
  \end{equation}
  for all $i \in \Ord$, since
  $\pfs(\lambda \alpha. \alpha_i(f))(S) = \{\alpha_i(f) \mid \alpha
  \in S\}$.

  Now, let $X = \{x,y\}$ and define
  $c^x \colon \fin{\pfs}^X \Rightarrow \fin{\pfs}$ by
  $c^x_i(f)= f(x)$, and similarly $c^y_i(f) = f(y)$. Further, define
  $d \colon \fin{\pfs}^X \Rightarrow \fin{\pfs}$ by
  \[d_i(f) = \begin{cases} f(x) & \text{ if } f(x) = \emptyset \\ f(y) & \text{ otherwise} \end{cases}\]
  for all $i \in \Ord$. It is easy to check that $c^x$ and $c^y$ are
  natural, and for $d$ this follows since the direct image of a set is
  empty iff the set itself is empty.
  % In fact, $c^x$, $c^y$ and $d$ arise as in
  % Lemma~\ref{lm:dl-alpha-gsos}
  %	from the distributive law presented by the GSOS specification
  %	$$
  %	\frac{x \xrightarrow{a} x'}{c_x(x,y) \rightarrow x'}
  %	\quad
  %	\frac{y \xrightarrow{a} y'}{c_x(x,y) \rightarrow y'}
  %	\quad
  %	\frac{x \xrightarrow{a} x' \quad y \xrightarrow{b} y'}{d(x,y) \rightarrow y'}
  %	$$
  %	but we do not go into details here.

  By the concrete characterisation in~\eqref{eq:hatepspow} and the
  definition of $c^x,c^y$ and $d$ we have
  ${(\widehat{\pfs\epsilon})}_X(\{c^x,c^y\}) = {(\widehat{\pfs\epsilon})}_X(\{c^x,c^y,d\})$. Hence,
  $\widehat{\pfs\epsilon}$ is not an isomorphism.
\end{proof}

\section{Related Work}%
\label{sec:rw}

The central notion of \emph{companion} proposed in the current paper
is a categorical generalisation of the lattice-theoretic notion, which
appeared first in~\cite{paco13} and was studied systematically
in~\cite{pous:lics16:cawu}. The construction of the companion
in terms of right Kan extensions generalises the lattice-theoretic
results of Parrow and Weber~\cite{ParrowWeber16,pous:lics16:cawu}.
To the best of our knowledge, our categorical notion of companion,
its structural properties, its construction as a right Kan extension
and the implications for causality are orginal contributions.

The current paper fits in the tradition of distributive laws in
universal coalgebra, started by Turi and Plotkin~\cite{TuriP97} and
subsequently extended in numerous papers.  The companion is
characterised as the final object in a category of distributive laws,
or, more generally, morphisms of endofunctors
(e.g.,~\cite{LenisaPW00}). Morphisms between these distributive laws
have been studied in various
papers~\cite{PowerW02,Watanabe02,LenisaPW00,KlinN15}, but a
final distributive law (the companion) does not appear there.

The current work is a thoroughly extended version of a conference
paper~\cite{PousR17}.  The new material includes proofs of all
results, examples, an introductory section on coinduction principles
(Section~\ref{sec:dl}), a new result on the powerset functor
(Section~\ref{sec:pres-pow}) and a generalisation from the companion
of a single functor to the companion of a pair $(A,B)$ of functors,
which allows to abstractly capture `causal operations' between
different final coalgebras (Section~\ref{sec:causal}).

The recent~\cite{BPR17} takes a different and more abstract approach:
there, the companion is constructed based on the theory of locally
presentable categories and accessible functors. That paper also
studies several other abstract properties of the companion, and
features higher-order companions.  The abstract constructions
of~\cite{BPR17} do not mention causality or its implications, which we
obtain here through the more explicit final sequence construction.  A
detailed comparison between the two constructions is left for future
work.

The use of distributive laws in enhanced (co)induction and % chktex 36
(co)iteration principles has been studied extensively, see, % chktex 36
e.g.,~\cite{Bartels04,Bartels03,LenisaPW00,UustaluVP01,Jacobs06,MiliusMS13}.
The notion of validity that we introduce in Section~\ref{sec:dl} comes
essentially from the work of Bartels~\cite{Bartels03} (where it is
phrased in terms of the algebra induced by a distributive law). It
slightly strengthens soundness in lattices, and also differs from the
notion of soundness in~\cite{bppr:acta:16}.  In contrast to the
above-mentioned approaches, our coinduction up-to principle for causal
algebras (Section~\ref{sec:causal-coinduction}) does not explicitly
refer to distributive laws.

Causality of stream functions is a well-established notion
(e.g.,~\cite{HansenKR16}).  To the best of our knowledge, the
generalisation to $\omega$-continuous functors in
Section~\ref{sec:causal}, and the construction in terms of the final
sequence, are new. The characterisation of causal operations in terms
of distributive laws generalises a known construction for
streams~\cite{HansenKR16} and automata~\cite{RotBR16} (modulo a
subtlety concerning syntax, see Section~\ref{sec:causal-alg-dl}). The
result that the contextual closure of any causal algebra lies below
the companion (and hence is sound) generalises a
soundness result of causal operations for
streams~\cite{EndrullisHB13} (the latter
also includes other up-to techniques, which we do not address here).

\section{Future work}%
\label{sec:fw}

As explained in Section~\ref{sec:pres-pow}, whether the finite
powerset functor has a companion remains open. This is important in
practice as this functor is used, e.g., to handle labelled transition
systems. This functor does not satisfy the main hypothesis of
Theorem~\ref{thm:companion-codensity}, but it could nevertheless be
the case that its codensity monad is its companion. If it exists,
one should understand its relationship with the family of
\emph{bounded companions} that one can obtain thanks to the
accessibility of the finite powerset functor~\cite{BPR17}. (Note that
the full powerset functor, which is not accessible, cannot have a
companion: this would entail existence of a final coalgebra.)

Another research direction consists in studying semantics of open
terms through the companion. Indeed, the present work shows that
causal operations on the final coalgebra can be presented as (plain)
distributive laws, to which the companion gives a canonical semantics
by finality. One should thus understand under which conditions the
usual notions of bisimulations on open terms coincide with this final
semantics, or whether generic notions of bisimulations can be designed
to capture it.

We discussed in Section~\ref{ssec:fibrations} how to combine the
present results with those from~\cite{bppr:acta:16}, where we use
fibrations to relate coinduction at the level of systems (e.g.,
defining streams corecursively), and coinduction at the level of
predicates and relations (e.g., proving equalities between streams coinductively, using
bisimulations). We should still understand precisely the final
sequence computation we use in the present work in that fibrational
setting (the work of Hasuo et al.\ on final sequences in a fibration~\cite{HasuoCKJ13}
may be a good starting point). In particular, we would like to understand the connection
between the companion of the lifted functor and the companions of the
induced functors on the fibres. This would make it possible to handle
other coinductive predicates than plain behavioural equivalence as in
Theorem~\ref{thm:causal-below-t}, and to understand the relationship
between the lattice-theoretic and coalgebraic companions.

\bibliographystyle{alpha}
\bibliography{companion-arxiv-rev}

\end{document}